\documentclass[%
 reprint,
nofootinbib,
 amsmath,amssymb,
 aps,
]{revtex4-1}
\usepackage{mathtools}
\usepackage[dvipsnames]{xcolor}

\usepackage{graphicx}
\usepackage{dcolumn}
\usepackage{bm}
\usepackage{xcolor}
\usepackage{multirow}
\usepackage{tikz}
\usepackage[break links]{hyperref}

\usepackage{graphicx}
\usepackage{comment}
\usepackage[all]{xy}
\usepackage{blkarray}

\usepackage[shortlabels]{enumitem}
\usepackage{amsthm}
\makeatletter
\newtheorem*{rep@theorem}{\rep@title}
\newcommand{\newreptheorem}[2]{%
\newenvironment{rep#1}[1]{%
 \def\rep@title{#2 \ref{##1}}%
 \begin{rep@theorem}}%
 {\end{rep@theorem}}}
\makeatother
\newtheorem{theorem}{Theorem}[section]
\newreptheorem{theorem}{Theorem}
\newtheorem{corollary}[theorem]{Corollary}
\newtheorem{prop}[theorem]{Proposition}
\newtheorem{Lemma}[theorem]{Lemma}
\newreptheorem{Lemma}{Lemma}
\newtheorem{claim}{Claim}

\theoremstyle{definition}
\newtheorem{defn}[theorem]{Definition}

\newcommand{\jhc}[1]{{\color{blue}  #1}}

\begin{document}

\preprint{APS/123-QED}

\title{Faking Gauge Coupling Unification in String Theory}

\author{James Halverson$^{1,2}$}
 \author{Benjamin Sung$^{2}$}%
 \affiliation{%
 \vspace{.2cm}
 $^1$The NSF AI Institute for Artificial Intelligence and Fundamental Interactions\vspace{.15cm}\\
 $^2$Department of Physics, \\ Northeastern University, \\Boston, MA 02115, USA
 }%

\begin{abstract}
    Gauge coupling unification misleads infrared observers if new gauge bosons do not simultaneously come into the spectrum. Though easy to engineer in gauge theory, the situation in string theory is nuanced, due to moduli dependence.
    We study the possibility of faking gauge coupling unification in the context of $4$d F-theory compactifications. Specifically, we formulate a sufficient condition that we call a \emph{strong calibration}, under which seven-brane gauge couplings on homologically distinct divisors become equal at codimension one in K\"{a}hler moduli space. We prove that a strong calibration is preserved under appropriate topological transitions and that a pair of non-intersecting divisors each admitting a contraction can always be strongly calibrated. Within the Tree ensemble~\cite{Halverson:2017ffz}, we find that $\approx 77.12\%$ of pairs of intersecting toric divisors can be strongly calibrated and $\approx 3.22\%$ can never be calibrated. Physically, this means that gauge coupling unification can be faked in most cases that we study, but in others
	it cannot, which is surprising from a gauge theoretic perspective.
\end{abstract}

\maketitle


\section{Introduction}

Grand unification is an enticing possibility for physics beyond the Standard Model \cite{georgiglashow,patisalam}. There are numerous lines of suggestive evidence for the hypothesis, such as gauge coupling unification and the representation theoretic structure of Standard Model (SM) fermions, which, together with an appeal to unity and beauty in particle physics, has driven an enormous amount of research \cite{langackerGUT}.  However, the simplest models introduce phenomenological problems such as rapid proton decay and GUT monopoles, neither of which have been observed, though they may be addressed with model building and inflation.


Alternatively, other hypotheses might account for the phenomena that are regularly cited as evidence for grand unification. For instance, if instead chiral gauge interactions are the governing principle, then the representation theoretic structure of the SM might instead be an accident of anomaly cancellation \cite{PhysRevD.39.693,PhysRevD.41.715,PhysRevD.41.717}, which often correlates with embedding into simple representations of larger groups.
This is a group theoretic explanation that does not require unification into a non-abelian gauge theory associated to the larger group. 
Similarly, the weak hypercharge is often cited as evidence for grand unification; but it is also the \emph{only} anomaly-free chiral $U(1)$ that can charge the $SU(3)\times SU(2)$ fermion content of the SM.

Which brings us to the subject of this article: gauge coupling unification. Should we find it compelling? 
On one hand, the Standard Model nearly exhibits
equality of gauge couplings at high scale, a necessary condition for grand unification,
and it is an impressive fact that weak-scale supersymmetry, a theory that is compelling for independent reasons, significantly improves the situation \cite{Langacker:1993ai,Ellis:1990wk,Langacker:1991an,Giunti:1991ta}. On the other hand, a skeptic might see the log-scale plot of $\alpha_i^{-1}$ and point out that, since two lines in $\mathbb{R}^2$ generically intersect in a point, having three requires tuning only one real parameter; i.e., high-scale equality of SM gauge couplings is a real codimension-one tuning. The grand unification proponent might counter that, yes, two lines generically intersect in a point, but not necessarily below the Planck scale! To which the skeptic gladly points out that this is an inequality, and is not even a codimension-one effect. 

\vspace{.3cm}
If the near gauge coupling unification observed in the Standard Model is an accident --- i.e., they do have similar values at high-scale, but there is no unification into a larger non-abelian theory --- one might say that gauge coupling unification is a fake that has misled the infrared observer toward grand unification. Achieving such a situation in gauge theory clearly requires a tuning, but perhaps it is mild enough to be acceptable, especially given the existence of other apparent tunings that are of much greator concern, such as the Higgs mass  and the cosmological constant; see, e.g., \cite{Weinberg:1975gm} and   \cite{Weinberg:1988cp}, respectively.

The situation is much more nuanced in string theory, however, since gauge couplings arise from scalar field expectation values. Generally, this moduli dependence can give rise to correlations between couplings and numbers of degrees of freedom that can affect low energy physics. For instance, in a different physical application, the $O(10^{15})$ compactifications with the exact chiral spectrum of the SM \cite{Cvetic:2019gnh} exhibit moduli dependence that yields a delicate interplay between controlling the theory, obtaining the correct value of the gauge couplings, dark gauge sectors, and ultralight axions \cite{Cvetic:2020fkd}. Notably, this ensemble automatically realizes gauge coupling unification without needing to tune moduli, due to homological properties satisfied by the SM sector.

In this paper we study the possibility of faking gauge coupling unification in string theory. Specifically, we study whether the moduli may be tuned such that two gauge couplings may be made equal. Relative to the above discussion, one should think of this as the string theoretic analog of picking a scale and then tuning two lines on the $\alpha_i^{-1}$ plot until they intersect. In gauge theory this tuning is clearly codimension one, and a central question for us is whether it is also codimension one in string theory. Our main result is that gauge couplings very often ($\approx 77.12\%$ in the Tree ensemble) become equal on a codimension one locus in moduli space, though there are also a small percentage ($\approx 3.22\%$ in the Tree ensemble) of cases in which it is not possible \emph{anywhere} in moduli space. The latter is surprising on its own, and is worthy of further study. See \cite{Gato-Rivera:2014afa} for other aspects of faking grand unification in string theory, and \cite{Dienes:1995bx,Dijkstra:2004cc} for early statistical studies of gauge coupling unification.

Our calculations will be carried out in the context of F-theory compactifications on an elliptic fibration $X\to B$ over a K\"ahler threefold $B$, the extra dimensions of space. In this context, gauge sectors arise from seven-branes wrapping divisors (four-cycles). Two seven-brane gauge sectors generally arise from two divisors, $D_i$ and $D_j$, and their gauge couplings are equal when $\text{vol}(D_i)=\text{vol}(D_j)$. These volumes are a function of K\" ahler moduli, and for homologically distinct divisors it is generally the case that the volumes are equal only on a sublocus in moduli space. We will quantify this according to a few conditions, especially a so-called \emph{strong calibration} that we define, which guarantees the existence of an equal volume locus at codimension one in K\" ahler moduli space. 
Furthermore, strong calibration is preserved under appropriate topological transitions; it is a property that may propagate through large networks of string geometries, akin to the ED3-instanton analysis of \cite{Halverson:2019vmd}.

This paper is organized as follows. In Section \ref{section:4ftheory} we introduce aspects of four-dimensional F-theory compactifications. In Section  \ref{sec:conditions} we define the study of the equal volume locus as a quadratic program, define strong calibration, and arrive at a number of general results related to the existence of equal volume loci. In Section \ref{section:calibration} we apply the general results to arrive at more specific conditions in the case that $B$ is a toric threefold. We also review the Tree ensemble, a colletion of $10^{755}$ topologically distinct toric threefolds $B$ that are related in a connected moduli space; this motivates the restriction to the toric case. Finally, in Section \ref{section:statistics} we apply those specific conditions to the Tree ensemble, demonstrating that, more often than not, there is an equal volume locus at codimension one that fakes gauge coupling unification, though sometimes the locus does not exist. We also give a concrete example exhibiting these ideas.

\section{4d F-theory}\label{section:4ftheory}

We will study the possibility of faking gauge coupling unification in the context of 4d F-theory compactifications.
The geometric nature of F-theory has led to the largest class of $4d$ $\mathcal{N}=1$ compactifications to date. In four-dimensional theories arising from compactification on an elliptically fibered Calabi-Yau fourfold $X \rightarrow B$, codimension $1,2,$ and $3$ intersections of irreducible components of the discriminant $\Delta \subset B$ lead to gauge algebras, matter, and Yukawa couplings respectively. 

We first specify consistent $4$d F-theory backgrounds. We recall the following:
\begin{defn}
An F-theory compactification geometry is specified by a smooth threefold base $B$ satisfying the Hayakawa-Wang criterion \cite{Hayakawa1995DEGENERATIONOC,Wang97onthe}: For any divisor $D \subset B$, there exist sections $f\in H^0(B,\omega_B^{\otimes-4}), g\in H^0(B,\omega_B^{\otimes-6})$  such that $mult_D(f,g) < (4,6)$. 
\end{defn}
\noindent This ensures that the associated Weierstrass elliptic Calabi-Yau fourfold has at worst canonical singularities~\cite{grassiminimal}, and hence is at finite distance from the bulk of the moduli space. The elliptic fourfold may be written in Weierstrass form
\[
y^2 = x^3 + fx + g, \qquad \Delta = 4f^3 + 27g^2 = 0
\]
$f \in \Gamma(-4K_B), g \in \Gamma(-6K_B)$.
Each reduced, irreducible component of $\Delta$ corresponds to a $7$-brane and the corresponding gauge algebra is specified by the order of vanishing of $\Delta$ along the corresponding component. In particular, such a $7$-brane is geometrically non-Higgsable if such a gauge algebra exists for generic sections $f$ and $g$; this is a gauge sector that cannot be removed by complex structure deformation.

In order to study F-theory geometries as broadly as possible, we must:
\begin{enumerate}[(a)]
\item
develop an efficient method to check the Hayakawa-Wang criterion for an arbitrary base,
\item
systematically understand the possible gauge algebras arising in the compactification.
\end{enumerate}
The history of this problem is rich and particularly well-understood in the case of base surfaces~\cite{Morrison_2012,Morrison_20121,Morrison_20122,taylor2017nontoric}, which yield compactifications to six dimensions. In this case, the questions a) and b) can be addressed precisely due to
\begin{enumerate}[(a)]
\item
finiteness results for elliptic threefolds~\cite{grassiminimal,1993alg.geom..5002G},
\item
Zariski decomposition for surfaces \cite{morrisontaylor2012}.
\end{enumerate}
In particular (b) allows one to determine the existence of gauge algebras associated with non-Higgsable clusters for general $6$d F-theory compactifications.

For 4d compactifications, there are fewer systematic results, but certain aspects of a) and b) may be addressed.
\begin{enumerate}[(a)]
	\item
	Finiteness of elliptic fourfolds (up to birational transformation) was established in recent work \cite{dicerbo2019birational}, using  \cite{10.2307/40587279}. 
	\item
	Though the gauge algebra may be checked in many concrete examples, the lack of a higher-dimensional analog of the Zariski decomposition prevents a systematic gauge algebra analysis for four-dimensional compactifications.
\end{enumerate}
However, since we will be working with toric bases, the Hayakawa-Wang criterion and the gauge algebra may be studied directly.

\section{Conditions for Faking Gauge Coupling Unification \label{sec:conditions}}

We turn to general conditions that are sufficient for faking GCU. By this we mean that there exists a sublocus in K\" ahler moduli space on which two divisors have equal volume, in which case we say they are \emph{calibrated} to the same volume\footnote{Though this English word is natural to describe the phenomenon, its use in this case is distinct from the related notion of calibrating the volume of a single calibrated submanifold \cite{10.1007/BF02392726}.}. Lemma \ref{lem:existence} provides a simple condition that is sufficient for the existence of such a sublocus, at codimension one; we call this a strong calibration. Lemmas \ref{lem:disjcalib} and \ref{lem:persistence} utilize strong calibration in particular cases; the former ensures that disjoint exceptional divisors may be calibrated, while the latter ensures the persistence of the calibratibility of two divisors under smooth blowups. We will demonstrate that both are useful for studying calibration in large ensembles of string geometries related by topological transitions. Finally, Lemma \ref{lem:calibexc} simplifies strong calibration in a case of interest to the ensemble.

\bigskip 

We begin by stating the equal volume condition, choosing to consider $\text{Nef}(X)$ in place of the K\"ahler cone so that results may be stated as intersections. 
Let $J_1,...,J_n$ denote the generators of $\text{Nef}(X)$ and fix $D_1$ and $D_2$, two effective divisors on $X$. We are interested in solutions to the system
\begin{equation}\label{eq:div3}
(a_1 J_1 + ... + a_n J_n)^2 \cdot (D_1 - D_2) = 0, \, a_i > 0\, \quad  \forall \, i
\end{equation}
which are loci in moduli space where $D_1$ and $D_2$ have equal volume.
Equivalently, define the $n\times n$ matrix $Q^{D_1 - D_2}$ with components $Q^{D_1 - D_2}_{ij} = J_{i} \cdot J_{j} \cdot (D_{1} - D_{2})$ and equation~\ref{eq:div3} is equivalent to considering existence of positive solutions to quadratic forms:
\begin{equation}\label{eq:div4}
a^T Q^{D_1-D_2} a = 0,\, a \in \mathbb{R}^n, \,a_i >0 \, \qquad\forall\, i
\end{equation}
and more generally for a collection of divisors $\{D_1,...,D_k\}$ existence of positive solutions to the system
\begin{align}\label{eq:div5}
a^T Q^{D_1-D_2}a &= 0,\,\,\, \dots\,\,\,  ,\,a^T Q^{D_{k-1}-D_k}a = 0, \nonumber \\ &a\in \mathbb{R}^n, \, a_i >0 \, \qquad \forall \, i
\end{align}
This type of problem is known as a quadratic program, which in general may be computationally complex. They arise in many places in string and field theory \cite{Halverson:2018cio}.

In the absence of an exact characterization of divisors admitting such a solution, we will reduce the question of existence of a solution to equation~\eqref{eq:div4} to a simpler constraint that is sufficient for the existence of a solution. We will then study particular cases related to blowups that are computationally feasible to check. 

Given two effective divisors $D_1, D_2$ on $X$, consider the corresponding quadratic form $Q^{D_1 - D_2}(a) = a^T Q^{D_1-D_2}a$. 
We define the positive orthants
\[\mathbb{R}^n_+ = \{(a_1,\ldots, a_n)\, |\, a_i \geq 0\, \forall\, i\}
\]
In particular, the interior of the positive orthant is given by $Int(\mathbb{R}^n_+) =  \{(a_1,\ldots, a_n)\, |\, a_i > 0\, \forall\, i\}$
\begin{Lemma}\label{lem:existence} 
Let $Q \colon \mathbb{R}^n \rightarrow \mathbb{R}$ be a quadratic form. Assume that there exists $a,b \in \mathbb{R}^n_+$ such that $Q(a) > 0$ and $Q(b) < 0$. Then there exists an open neighborhood $U \subset Int(\mathbb{R}^n_+)$ such that $U \cap Q^{-1}(0)$ is a submanifold of codimension $1$.
\end{Lemma}
\begin{defn}[Strong Calibration]
Given two divisors $D_1,D_2 \subset X$, we say that the pair $(D_1,D_2)$ can be \emph{strongly calibrated} if it satisfies the assumptions of Lemma~\ref{lem:existence}, since the assumptions are sufficient but not necessary for calibration.
\end{defn}

\noindent In principle, two divisors $D_1,D_2$ can potentially be calibrated even if they do not admit a strong calibration; i.e. there may exist $J$ ample on $X$ such that $J^2 \cdot D_1 = J^2 \cdot D_2$, but not under the assumptions of Lemma \ref{lem:existence}.

A natural concern is whether the neighorbood $U$ of Lemma \ref{lem:existence} is in a regime of control of the theory. Let $p$ be the center of $U$, and suppose that it is in the K\" ahler cone, but not the stretched K\" ahler cone \cite{Demirtas:2018akl}. Then the gauge coupling may not simply be a divisor volume, due to $\alpha'$ corrections, potentially spoiling volume-related analyses of gauge coupling unification. However, $p$ may be scaled out by a multiplicative factor, yielding a new point $p'$ satisfying the same conditions that is arbitrarily deep inside the stretched K\" ahler cone. The existence of gauge coupling unification via calibration may therefore be studied reliably, though obtaining a specific \emph{value} of the unified coupling may not be possible in a regime of control. Such interplays between the stretched K\" ahler cone and correlations between gauge couplings and other physical observables were studied in \cite{Cvetic:2020fkd}.

\subsection{Existence and persistence under blowups}
In this section, we explore circumstances under which two divisors can be calibrated. We will primarily be interested in the following setting: Assume $\pi \colon \hat{X} \rightarrow X$ a blowup of smooth complex projective varieties with smooth center. We denote by $E$ the exceptional divisor. Given two divisors $D_1,D_2 \subset X$, we denote by $\widetilde{D}_1, \widetilde{D}_2$ their proper transforms.

We first demonstrate that a calibration always exists assuming that $D_1$ and $D_2$ are disjoint and each admit contractions. 
\begin{Lemma}\label{lem:disjcalib}
Assume in addition that $D_1$ and $D_2$ are disjoint and there exists  birational contractions $\pi_i \colon X \rightarrow X_i, i \in \{1,2\}$, such that $D_i$ is the exceptional divisor of $\pi_i$. Then $D_1$ and $D_2$ can be calibrated.
\end{Lemma}
\begin{proof}
Without loss of generality, we consider the morphism $\pi_1 \colon X \rightarrow X_1$. As $\pi_1$ is an isomorphism away from $D_1$, $\pi_1(D_2) \subset X_1$ is an effective divisor and $\pi_1^*\pi_1(D_2) = D_2$. As $X_1$ is projective, there exists an ample class $J$ on $X_1$. In particular, we have that $(\pi_1^*J)^2 \cdot D_2 = (\pi_1^*J)^2 \cdot \pi_1^* \pi_1(D_2) = J^2 \cdot \pi_1( D_2) > 0$ and $(\pi_1^*J)^2 \cdot D_1 = \pi_1^*J^2\cdot D_1 = J^2 \cdot \pi_1(D_1) = 0$ by proposition~\ref{prop:intfacts}. Thus, we have that $(\pi_1^*J)^2 \cdot (D_2 - D_1) > 0$. The opposite inequality follows from the same argument applied to $\pi_2$. 
\end{proof}

Moreover, the existence of a calibration is always preserved under smooth blowups. 

\begin{Lemma}\label{lem:persistence}
Assume that $D_1, D_2$ are divisors on $X$ that can be strongly calibrated. Then their proper transforms $\widetilde{D}_1, \widetilde{D}_2$ can be calibrated on $\hat{X}$.
\end{Lemma}
\begin{proof}
The classes of the proper transforms satisfy $\widetilde{D_i} = \pi^* D_i - a_iE $ for some $a_i$. By assumption, there exists $J_1,J_2 \in Nef(X)$ such that $J_1^2 \cdot (D_1 - D_2) > 0 $ and $J_2^2 \cdot (D_1 - D_2) < 0$. In particular, we have that $(\pi^*J_i)^2 \cdot (\widetilde{D}_1 - \widetilde{D}_2) = (\pi^*J_i)^2 \cdot (\pi^* D_1 - \pi^* D_2) = J_i^2 \cdot (D_1 - D_2) $ by proposition~\ref{prop:intfacts}. Moreover, as the pullbacks of nef classes are nef by corollary~\ref{cor:intfacts}, the divisors $\widetilde D_1, \widetilde D_2$ can be calibrated.
\end{proof}
\subsection{Calibration of exceptional divisors}
Our primary interest will be in the calibration of exceptional divisors arising from blowups of smooth points and curves. In this section, we remark on a Lemma which simplifies the condition of Lemma~\ref{lem:existence} in our settings of interest. In addition, we record some facts about intersections with exceptional divisors in Appendix \ref{app:intersections} which will be used extensively later. 

Assume $\pi \colon Y \rightarrow X$ a blowup of a smooth complex projective variety with smooth center. Let $E$ be the exceptional divisor. Then in a common case we have a simplification of Lemma \ref{lem:existence}:
\begin{Lemma}\label{lem:calibexc}
Let $\widetilde{D}$ the proper transform of any divisor $D \subset X$. Then the pair $(\widetilde{D},E)$ admits a calibration if and only if there exists a nef class $J \in Nef(Y)$ such that $J^2\cdot (E - \widetilde{D}) > 0$.
\end{Lemma}
\begin{proof}
Let $A$ be an ample class in $X$ and $\pi^* A$ its pullback. Then $(\pi^*(A))^2 \cdot (E-\tilde D) = - A^2 \cdot D < 0$ since $(\pi^*(A))^2 \cdot E = 0$.
\end{proof}
\noindent In other words, the opposite inequality in Lemma \ref{lem:existence} is automatic under the assumptions.

\section{Existence of calibration}\label{section:calibration}
In this section, we review the construction of a large class of $4d$ F-theory bases introduced in~\cite{Halverson:2017ffz}. Within the framework of this construction, we give conditions under which many exceptional divisors can indeed be strongly calibrated. Although we work in the framework of toric geometry for ease of illustration, all our arguments generalize to arbitrary smooth complex projective varieties exhibiting a similar pattern of blowups.
\subsection{Review of the tree ensemble}
 By restricting to blowups of weak Fano toric threefolds, we may systematically check the Hayakawa-Wang criterion and compute a lower bound on the geometric gauge algebra appearing in such a compactification. We thus have the following claims:
\begin{claim}[\cite{Halverson:2017ffz}]
There exists an algorithmic construction yielding a lower bound of $\frac{4}{3} \times 2.96 \times 10^{755}$ F-theory geometries with toric threefold basess, connected by topological transitions in a connected moduli space.
\end{claim}
\begin{corollary}
Drawing from a uniform distribution on the finite site of geometries, geometrically Non-Higgsable clusters occur with probability $\gtrapprox 1-10^{-755}$; they are a universal feature of the construction.
\end{corollary}
\noindent These so-called geometrically non-Higgsable clusters are F-theory compactifications over K\" ahler bases that exhibit non-abelian gauge symmetry on seven-branes for generic complex structure moduli, i.e., there are no complex structure deformations that can break the symmetry. They were first realized in 6d compactifications in \cite{Morrison:1996na, Morrison:1996pp}, but were not studied extensively for a number of years, until \cite{Morrison:2012js,morrisontaylor2012,Grassi:2014zxa,Morrison:2014lca,Halverson:2015jua,Halverson:2016vwx} uncovered many of their properties. They are generic features of 4d F-theory compactifications when drawing from a uniform distribution on large ensembles of F-theory geometries
\cite{Taylor:2015ppa,Halverson:2017ffz,Taylor:2017yqr}
of flux vacua \cite{Taylor:2015xtz}; see \cite{Halverson:2016nfq,Halverson:2018olu,Halverson:2019kna,Halverson:2019cmy,Halverson:2020xpg} for cosmological implications of non-Higgsable clusters. The claim and corollary are sharp lower bounds that arise in the ensemble we study.

\vspace{.2cm}
We briefly review the construction of \cite{Halverson:2017ffz}. Let $B$ be a smooth algebraic threefold, and $\hat{B} \rightarrow B$ a blowup with smooth center. Then given a generic Weierstrass Calabi-Yau fourfold $\pi \colon X \rightarrow B$, there is an induced Weierstrass Calabi-Yau fourfold $\hat{X} \rightarrow \hat{B}$ obtained by a base change of $\pi$ together with a change of coordinates. 

In order to compute the relevant physical quantities and to verify the Hayakawa-Wang criterion explicitly, we restrict to the case where $B$ is toric. We then probe the space of all elliptic Calabi-Yau fourfolds by studying smooth toric blowups of weak Fano toric threefolds and their induced Weierstrass models. By leveraging the power of toric geometry, we will rephrase all the relevant sufficient geometric conditions in terms of simple combinatorics.

We review the construction of \cite{Halverson:2017ffz}, which introduced the notions of trees and leaves below. Let $B$ be a smooth weak Fano toric threefold. This is equivalent to specifying a fine, regular star triangulation (FRST) of one of $4319$ $3$d reflexive polytopes \cite{Kreuzer:1998vb}. A toric point or curve corresponds to a $3$d or $2$d simplex, respectively, specified by a sequence of vertices $(v_1,v_2,v_3)$ or $(v_x,v_y)$ of the triangulation. A blowup of $B$ at a toric point is specified by replacing the simplex $(v_1,v_2,v_3)$ with the set of simplices $(v_1,v_2,v_f),(v_2,v_3,v_f),(v_1,v_3,v_f)$ where $v_f = v_1 +v_2 +v_3$. An edge $(v_x,v_y)$ has two unique neighboring vertices $(a,b)$ such that $(a,b,v_x),(a,b,v_y)$ form $3$-simplices of the triangulation. Similarly, a blowup of $B$ at a toric curve is specified by replacing the simplices $(a,v_x,v_y),(b,v_x,v_y)$ with the four simplices $(a,v_x,v_e),(a,v_y,v_e),(b,v_x,v_e),(b,v_y,v_e)$,  $v_e = v_x+v_y$. 

Given a $3$-simplex $f = (v_1,v_2,v_3)$, a \emph{face tree} on $f$ is specified by first blowing up by adding the vertex $v_f = v_1 +v_2+v_3$ and subdividing, and then performing arbitrary compositions of face and edge blowups not including edge blowups of the bounding three edges. Similarly, given a $2$-simplex $e = (v_x,v_y)$, an \emph{edge tree} is specified by arbitrary compositions of edge blowups. A \emph{leaf} of a face (resp. edge) tree is an added vertex $v$ which can be uniquely written as $v = a_1 v_1 + a_2 v_2 + a_3 v_3$ (resp. $v = a_x v_x + a_y v_y$) with coefficients $a_i \geq 1$. The \emph{height} of a leaf $v$ of a face (resp. edge) tree is the sum $a_1 + a_2 + a_3$ (resp. $a_x + a_y$). 

An arbitrary F-theory base in our ensemble will be specified by an FRST triangulation of a $3$d reflexive polytope, a face tree for every $3$-simplex, and then an edge tree for every $2$-simplex. The following picture illustrates the latter part of this procedure by projecting the toric blowup data into a facet of the original triangulated polytope: on the left, only the face trees corresponding to the green edges are added, and on the right, an edge tree has been added to the common $2$-simplex. Moreover the corresponding divisor classes are given by the following:
\begin{enumerate}
\item
Before adding trees, $D_x,D_y$ denote divisor classes at intersection of two simplices.
\item
After adding green leaves:
\begin{align*}
\widetilde{D_x} &= \pi^*D_x - \pi^*F_2 - \pi^*F_3 - \pi^*F_1\\
\widetilde{D_y} &= \pi^*D_y - \pi^*F_1 - \pi^*F_2
\end{align*}
\item
After adding red leaves:
\begin{align*}
\widetilde{\widetilde{D_x}}&= \pi^*\widetilde{D_x} - \pi^*E_1 - E_2 \\
\widetilde{\widetilde{D_y}} &= \pi^*\widetilde{D_y} - \pi^*E_1
\end{align*}
\end{enumerate}
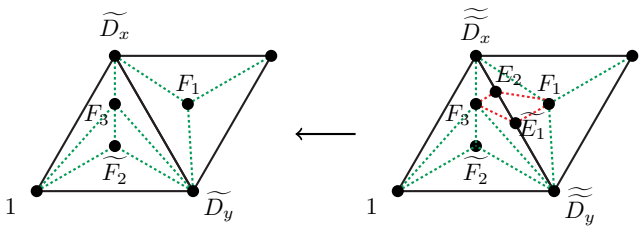
\begin{center}
\begin{tikzpicture}[scale=1.6, every node/.style={scale=.9}]
\draw[thick,color=Black] (90:.75) -- (90+120:.75) -- (90+120+120:.75) -- cycle;
\draw[thick,color=Black] (90:.75) -- (30:1.50) -- (90+240:.75) -- cycle;
\draw[thick,dash pattern={on 1pt off 1pt},color=ForestGreen] (90:.75) -- (0,0);
\draw[thick,dash pattern={on 1pt off 1pt},color=ForestGreen] (90+120:.75) -- (0,0);
\draw[thick,dash pattern={on 1pt off 1pt},color=ForestGreen] (90+240:.75) -- (0,0);
\draw[thick,dash pattern={on 1pt off 1pt},color=ForestGreen] (90+240:.75) -- (90:.35);
\draw[thick,dash pattern={on 1pt off 1pt},color=ForestGreen] (90+120:.75) -- (90:.35);
\draw[thick,dash pattern={on 1pt off 1pt},color=ForestGreen] (90+240:.75) -- (30:.7);
\draw[thick,dash pattern={on 1pt off 1pt},color=ForestGreen] (90:.75) -- (30:.7);
\draw[thick,dash pattern={on 1pt off 1pt},color=ForestGreen] (30:.7) -- (30:1.5);
\fill (30:.7) circle (.5mm);
\fill (90:.75) circle (.5mm);
\fill (90:.35) circle (.5mm);
\fill (90+120:.75) circle (.5mm);
\fill (90+240:.75) circle (.5mm);
\fill (30:1.5) circle (.5mm);
\fill (0,0) circle (.5mm);
\node at (40:.8) {$F_1$};
\node at (120:.3) {$F_3$};
\node at (0,-.2) {$\widetilde{F_2}$};
\node at (90:1) {$\widetilde{D_x}$}; \node at (90+120:1) {$1$}; \node at (90+240:1) {$\widetilde{D_y}$};
\draw[thick,<-] (1.5,.1) -- (2,.1);
\begin{scope}[xshift=3cm]
\fill (0,0) circle (.5mm);
\draw[thick,color=Black] (90:.75) -- (90+120:.75) -- (90+120+120:.75) -- cycle;
\draw[thick,color=Black] (90:.75) -- (30:1.50) -- (90+240:.75) -- cycle;
\draw[thick,dash pattern={on 1pt off 1pt},color=ForestGreen] (90:.75) -- (0,0);
\draw[thick,dash pattern={on 1pt off 1pt},color=ForestGreen] (90+120:.75) -- (0,0);
\draw[thick,dash pattern={on 1pt off 1pt},color=ForestGreen] (90+240:.75) -- (0,0);
\draw[thick,dash pattern={on 1pt off 1pt},color=ForestGreen] (90+240:.75) -- (90:.35);
\draw[thick,dash pattern={on 1pt off 1pt},color=ForestGreen] (90+120:.75) -- (90:.35);
\draw[thick,dash pattern={on 1pt off 1pt},color=ForestGreen] (90+240:.75) -- (30:.7);
\draw[thick,dash pattern={on 1pt off 1pt},color=ForestGreen] (90:.75) -- (30:.7);
\draw[thick,dash pattern={on 1pt off 1pt},color=ForestGreen] (30:.7) -- (30:1.5);
\draw[thick,dash pattern={on 1pt off 1pt},color=Red] (90:.35) -- (30:.38);
\draw[thick,dash pattern={on 1pt off 1pt},color=Red] (30:.75) -- (30:.38);
\draw[thick,dash pattern={on 1pt off 1pt},color=Red] (70:.48) -- (90:.35);
\draw[thick,dash pattern={on 1pt off 1pt},color=Red] (70:.48) -- (30:.75);
\fill (90:.75) circle (.5mm);
\fill (90+120:.75) circle (.5mm);
\fill (90+240:.75) circle (.5mm);
\fill (90:.35) circle (.5mm);
\fill (30:1.5) circle (.5mm);
\fill (30:.7) circle (.5mm);
\fill (70:.48) circle (.5mm);
\fill (30:.38) circle (.5mm);

\node at (65:.65) {$E_2$};
\node at (20:.50) {$\widetilde{E_1}$};
\node at (40:.8) {$F_1$};
\node at (120:.3) {$F_3$};
\node at (0,-.2) {$\widetilde{F_2}$};
\node at (90:1) {$\widetilde{\widetilde{D_x}}$}; \node at (90+120:1) {$1$}; \node at (90+240:1) {$\widetilde{\widetilde{D_y}}$};

\end{scope}

\end{tikzpicture}
\end{center}

In this figure we denote exceptional divisors in face trees and edge trees with $F$'s and $E$'s, respectively, which gain subscripts when there are multiple such leaves. In the following they may also gain tildes when they are pulled back under subequent blowup. We utilized this notation for the remainder of the sections that utilize toric geometry.

In the subsequent sections, we will consider the existence of a calibration for any pair of intersecting toric divisors. In section~\ref{curves}, we prove that the toric divisors of the above form intersecting along the red lines and along segments on the common $2$-simplex can always be strongly calibrated. In section~\ref{points}, we reduce the question of the existence of a calibration for exceptional toric divisors of the above form intersecting along green lines to blowups of $\mathbb{P}^3$ and formulate a statement for the non-existence of a calibration of such divisors.

\subsection{Blowups of toric curves}\label{curves}
We characterize pairs of toric divisors that we claim can always be strongly calibrated. Given a weak Fano toric variety $X$, we will use the following notation. Fix $\pi_f \colon X_f \rightarrow X_{f-1} \rightarrow \ldots \rightarrow X_1 \rightarrow X$, the sequence of blowups establishing a single face tree and $\pi_e \colon X_e \rightarrow X_{e-1} \rightarrow \ldots \rightarrow X_{f+1} \rightarrow X_f \rightarrow X$, the sequence of blowups establishing a single edge tree on an adjacent toric curve $C \subset X$; i.e. a toric curve in $X$ that intersects the point that is blown up to establish the face tree. For each morphism in the composition $\pi_f$, we denote by $F_i$ the face leaf or exceptional divisor of the morphism $X_i \rightarrow X_{i-1}$ and for each morphism in the composition $\pi_e$, we denote by $E_j$ the edge leaf or exceptional divisor of the morphism $X_j \rightarrow X_{j-1}$.

For simplicity, given a map $\pi_m \colon X_m \rightarrow X_n \rightarrow X_p$ with $p < n < m$, a composition of blowup maps, and algebraic cycles $B_n \in A^*(X_n), B_p \in A^*(X_p)$, we denote the pullbacks to $A^*(X_m)$ under the truncation $X_m \rightarrow X_n$ and $\pi_m \colon X_m \rightarrow X_n \rightarrow X_p$ with the common notation $\pi_m^*B_n, \pi_m^*B_p$, respectively. In such a situation, we will specify which variety in the chain each cycle is pulled back from. Similarly, we will use the common notation $\widetilde{B_n},\widetilde{B_p}$ and specify that we are considering the proper transforms on $X_m$ under the given map and its truncation.

\begin{Lemma}\label{lem:redcurves}
The proper transforms $\widetilde{F}_j$ and $\widetilde{E}_k$, i.e. any pair of divisors associated to a pair of leaves in a face tree and an edge tree, repsectively, can always be strongly calibrated on $X_n$.
\end{Lemma}

\begin{Lemma}\label{lem:edgecurves}
Assume that a pair of proper transforms $\widetilde{E_j}, \widetilde{E_k}$ 
in the same edge tree intersect. Then the pair admits a strong calibration.
\end{Lemma}
\section{Blowups of points}\label{points}
Assume $X$ a smooth complex projective variety. Assume $X_0 \rightarrow X$ a blowup at a point $p$ with exceptional divisor $E_0$. Let $\pi \colon X_f \xrightarrow{\pi_f} X_{f-1} \xrightarrow{\pi_{f-1}} \ldots \xrightarrow{\pi_1} X_0 \xrightarrow{\pi_0} X$ be a sequence of smooth blowups such that the center of each blowup $\pi_i$ is contained in the exceptional locus $Exc(\pi_{i-1} \circ \ldots \circ \pi_1 \circ \pi_0)$.

Locally in the analytic topology, we may identify an open $U_p \subset X$ containing $p$ with an open neighborhood in $\mathbb{C}^n$. The above sequence of blowups induces a sequence of smooth blowups $\pi_U \colon U_f \rightarrow U_{f-1} \rightarrow \ldots \rightarrow U_0 \rightarrow U_p$. Via the canonical embeddings $U_p \subset \mathbb{C}^n \subset \mathbb{P}^n$, $\pi_U$ is obtained by a base change of a sequence of smooth blowups $\pi_{P} \colon P_f \rightarrow P_{f-1} \rightarrow \ldots \rightarrow P_0 \rightarrow \mathbb{P}^n$. Moreover, for any exceptional divisor $E \subset Exc(\pi)$ on $X_f$, we can identify $E$ with its image $E_P \subset P_f$ and vice versa.

\begin{Lemma}\label{lem:localblow}
The following holds for the divisor $J = \pi^* D+ E$, $D$ a divisor on $X$.
\begin{enumerate}
\item
$E$ is $\pi$-nef if and only if $\pi^* D + E$ is nef for some divisor $D \in N^1(X)$.
\item
Assume $E'\subset Exc(\pi)$ is a divisor on $X_f$. Then $(\pi^* D + E)^2 \cdot E' = E^2 \cdot E'$.
\end{enumerate}
\end{Lemma}
\begin{proof}
The if direction in Part 1 follows by definition and the projection formula \ref{prop:intfacts}(3). The only if direction follows from lemmas~\ref{lem:relativeJ} and \ref{lem:relativeJblow}.
Part 2 follows from \ref{prop:intfacts}(3).
\end{proof}

\begin{corollary}\label{cor:local}
Assume $E_1, E_2$ are divisors on $X_f$ in the exceptional locus $Exc(\pi)$. Then $E_1$ and $E_2$ can be strongly calibrated if and only if $E_1$ and $E_2$ can be strongly calibrated on $P_f$. In particular $E_1$ and $E_2$ can be calibrated if and only if $E_1$ and $E_2$ can be calibrated on $P_f$.
\end{corollary}
\begin{proof}
It suffices to prove that the assertion is local on $X$. By definition, $E_1$ and $E_2$ can be strongly calibrated if and only if there exist nef divisors $J_1, J_2 \in Nef(X_n)$ such that $J_1^2 \cdot E_1 > J_1^2 \cdot E_2$ and $J_2^2 \cdot E_1 < J_2^2 \cdot E_2$. We may assume $J_1 = \pi^* D_1 + F_1$ for some $D_1 \in N^1(X)$, $F_1 \subset Exc(\pi)$ and $E_1, E_2 \in Exc(\pi)$. Then $J_1$ is nef if and only if $F_1$ is $\pi$-nef by Lemma~\ref{lem:localblow}(1), and the first condition holds if and only if $F_1^2 \cdot E_1 > F_1^2 \cdot E_2 $ by Lemma~\ref{lem:localblow}(2).

Arguing similarly for the second condition, in summary, $E_1$ and $E_2$ can be strongly calibrated if and only if there exists $\pi$-nef divisors $F_i$ such that $F_1^2 \cdot E_1 > F_1^2 \cdot E_2$ and $F_2^2 \cdot E_1 < F_2^2 \cdot E_2$; note also that $F_i \subset Exc(\pi)$. But these divisors are all contained in $U_f$ and we conclude.
\end{proof}

For simplicity, we restrict to the case of toric varieties and the blowup construction in the tree ensemble. Under certain assumptions, the property of non-existence of an ample divisor $J$ such that $J^2 \cdot (D_1 - D_2) = 0$ can also be preserved. 
\begin{Lemma}\label{lem:nonexistence}
Assume $X_n \rightarrow \ldots \rightarrow X_m$ a sequence of blowups establishing a collection of edge trees. Assume that $J^2 \cdot (D_1 - D_2) > 0$ for any ample class $J \in Nef(X_m)$ with $D_1,D_2$ toric divisors on the same face tree. Assume in addition that $D_1$ does not intersect with any base curve of an edge tree. Then any ample class $J' \in Nef(X_n)$ satisfies $J'^2 \cdot (\widetilde{D_1} - \widetilde{D_2}) > 0$ on $X_n$.
\end{Lemma}

\section{Statistics in the ensemble}\label{section:statistics}
In this section, we apply the above results to compute the statistics of existence of calibrations within the tree ensemble. We also present an explicit example that demonstrates our Lemmas regarding the existence and non-existence of calibration.
\subsection{Faking Gauge Coupling Unification \\ in the Tree Ensemble}\label{subsection:intersecting}
We will apply corollary~\ref{cor:local} for each face tree in the ensemble
\[
\pi_P \colon P_f \rightarrow P_{f-1} \rightarrow \ldots \rightarrow P_0 \rightarrow \mathbb{P}^3,
\]
where we take $\mathbb{P}^3$ as the initial variety, w.l.o.g. by \ref{cor:local}.
For each face tree, we compute the relative nef cone $Nef(P_f/\mathbb{P}^n)$ consisting of nef divisors on $P_f$ relative to toric curves contracted by $\pi_P$. For each pair of intersecting toric exceptional divisors $E_1,E_2 \in Exc(\pi_P)$, if there exists generators $J_1, J_2 \in Nef(P_f/\mathbb{P}^n)$ such that $J_1^2 \cdot (E_1 - E_2)> 0$ and $J_2^2 \cdot (E_2 - E_1) > 0$, then $E_1$ and $E_2$ can be strongly calibrated on any identical face tree on an arbitrary smooth complex projective variety $X$. 

On the other hand, if $J_i \cdot J_j \cdot (E_1 - E_2) > 0 $ for all $i,j$, then $E_1, E_2$ can never be calibrated. Indeed, by the same argument as in corollary~\ref{cor:local}, it suffices to check that every element $J \in Amp(P_f/\mathbb{P}^3)$ satisfies $J^2 \cdot (E_1- E_2) > 0$ up to switching $E_1$ and $E_2$. But with our assumptions, we may assume $J = a_1 J_1 + \ldots + a_n J_n$ for $J_i$ all the generators of the cone $Nef(P_f/\mathbb{P}^3)$ and $a_i >0$ for all $i$. Then clearly, the assumption of the first sentence implies the claim.

For each tree, we compute the number of intersecting divisors admitting and not admitting a calibration. This yields a probability of calibration and not calibration for the intersecting pairs of divisors in each tree, and we average this over all trees in the ensemble, displaying the results in Table \ref{tab:summary}. Since our conditions are only sufficient for the existence or non-existence of calibration, the probabilities need not add to one.
\begin{table*}[t]
\centering
\begin{tabular}{c|c|c|c|c|c|c}
	\# Simplices & \# Face Trees & Calib Total & No Calib Total & Total Curves& P(Calib) & P(No Calib) \\
	 \hline
	3& 1 &  0 &0 & 3 & 0 & 0  \\

	5  & 6 & 6& 0 &36& 0.1666 & 0 \\
	7  & 33& 66& 0 &297 &0.2222 & 0 \\
	9  & 145& 525& 24 &1740 &0.3017 &0.0138 \\
	11  & 564 & 3126& 162&8460& 0.3695 & 0.0191\\
	13  & 2004& 15342&870 & 36072 & 0.4253 &0.0241 \\
	15  & 6586 &65166 & 3966&  138306 & 0.4712 & 0.0287\\
	17  & 20175 &246120 & 16452&484200 & 0.5083 & 0.0340 \\
	19  & 57729&841062 & 61461 &1558683 & 0.5396 & 0.0394\\
	21  &154069 &2618010 &206655 &4622070 & 0.5664 & 0.0447\\
	23  & 382206& 7438689& 621837&12612798 & 0.5898 & 0.0493 \\
	25  &876186 &19250893 & 1669752&31542696 & 0.6103 & 0.0529\\
	27  & 1839392 & 45096804& 3984672& 71736288& 0.6286 & 0.0555\\
	29  & 3485172 & 94452598&8357028 &146377224 & 0.6453 & 0.0571 \\
	31  & 5820060&173009204 &15185328& 261902700 & 0.6606 &0.0580 \\
	33  &8272334&268091974 & 23100966& 397072032 & 0.6752 & 0.0582\\
	35  & 9449160 &332325497 &  27686928& 481907160& 0.6896& 0.0574\\
	37  &7844748&299082419 &22680936 & 423616392 & 0.7060 &0.0535 \\
	39  & 3663074& 152297768&8237964 &208795218 & 0.7294 & 0.0395\\
	  \end{tabular}
\caption{Second column is the number of all face trees with specified number of simplices. Third and fourth columns are the total number of curves (intersecting pairs of toric divisors) admitting and not admitting a calibration over all face trees of the specified number of simplices respectively. Fifth and sixth columns are obtained by dividing the second and third columns by the fourth column respectively.}
\label{tab:summary}
\end{table*}

\begin{table}[t]
\centering
\begin{tabular}{c|c|c|c|c}
	\# Simplices & \# Edge Trees & Calib & Total Curves &P(Calib)\\
	 \hline

	1& 1 &  0  & 1& 0\\

	2  & 1 & 2 &4 & 0.5000\\
	3  & 2& 10 & 14&0.7142\\
	4  & 5& 40& 50 & 0.8000\\
	5  & 8 & 88& 104& 0.8462\\
	6  & 12& 168 &192 &0.8750\\
	7  & 14 &238 &266&0.8947\\
	8  & 14 & 280& 308 & 0.9090\\
	9  & 12& 276& 300&0.9200\\
	10  & 8 & 208 & 224&0.9286\\
	11  & 4 & 116&124&0.9355\\
	12  & 1 & 32 &34&0.9412\\
	  \end{tabular}
\caption{Table of all possible edge trees. We use the convention that a curve of an edge tree is any curve on top of the base edge, or any curve whose corresponding 1-simplex has endpoint interior to the base edge. The total curves is computed with the formula $(3(n-1)+1)\times N_e(n)$ where $n$ is the number of simplices and $N_e(n)$ is the number of corresponding edge trees.}
\label{tab:summaryedge}
\end{table}

To compute the statistics of a calibration for face trees, we operate with the following ansatz. Consider the set $\mathcal{T}_f$ of all ordered pairs $(T_f, \{D_i,D_j\})$ where $T_f$ is a face tree, and $D_i,D_j$ are two distinct intersecting toric divisors on $T_f$. We compute the following sums
\begin{align}
\begin{split}
E_f(\text{Calib}) &= \frac{1}{|\mathcal{T}_f|}\sum\limits_{t\in\mathcal{T}_f}P_c(t)\\
E_f(\text{No Calib})&= \frac{1}{|\mathcal{T}_f|}\sum\limits_{t\in\mathcal{T}_f}P_{nc}(t)
\end{split}
\end{align}
where $P_c(t) = 1$ if $D_i$ and $D_j$ admit a strong calibration and $0$ otherwise and $P_{nc}(t) = 1$ if $D_i$ and $D_j$ satisfy the assumptions of Lemma~\ref{lem:nonexistence} and $0$ otherwise. Summing over the third, fourth, and fifth columns of table~\ref{tab:summary}, we obtain the results:
\begin{align}\label{eq:face}
\begin{split}
\sum\limits_{t\in\mathcal{T}_f}P_c(t) = \text{Calib Total} &= 1394835269\\
\sum\limits_{t\in\mathcal{T}_f}P_{nc}(t) = \text{No Calib Total} &= 111815001\\
|\mathcal{T}_f| = \text{Total Curves} &= 2042412375
\end{split}
\end{align}
obtaining the results
\begin{align*}
E_f(\text{Calib}) &= 0.6829 \\
E_f(\text{No Calib}) &= 0.0547
\end{align*}
for face trees.

Similarly, we define the set $\mathcal{T}_e$ of all ordered pairs $(T_e,\{D_i,D_j\})$ where $T_e$ is an edge tree and $D_i,D_j$ are two intersecting divisors. Defining everything analogously and summing over the relevant data in table~\ref{tab:summaryedge}, we find
\begin{align}\label{eq:edge}
\begin{split}
\sum\limits_{t\in\mathcal{T}_e}P_c(t) &= 1458\\
|\mathcal{T}_e| &= 1621
\end{split}
\end{align}
for edge trees.

To estimate the statistics across the whole ensemble, we will study a fixed facet $F$ of a $3$d reflexive polytope with $36$ faces and $63$ edges; there are two such facets, which dominate the statistics, and each arises in a K\" ahler threefold with $h^{1,1}=35$. 

Let $\mathcal{T}$ be the set of all ordered pairs $(T,\{D_i,D_j\})$ where $T$ is a pair consisting of a sequence of $36$ face trees and a sequence of $63$ edge trees, and $D_i, D_j$ are two intersecting toric divisors contained in $T$. We compute the sums
\begin{align*}
E(\text{Calib}) &= \frac{1}{|\mathcal{T}|} \sum\limits_{t \in \mathcal{T}} P_c(t) \\
E(\text{No Calib}) &= \frac{1}{|\mathcal{T}|} \sum\limits_{t \in \mathcal{T}} P_{nc}(t) 
\end{align*}

To compute the quantity $|\mathcal{T}|$, we first define $\mathcal{F} \subset \mathcal{T}$ to be the subset of pairs whose intersecting divisors belong to the same face tree. Similarly, we define $\mathcal{E}\subset \mathcal{T}$ to be the subset of pairs with intersecting divisors belong to the same edge tree. We define $\mathcal{F}_f,\mathcal{E}_e$ to be the subsets of $\mathcal{F},\mathcal{E}$ consisting of pairs whose intersecting divisors belong to a face or edge tree lying over a fixed face $f$ or edge $e$ in $F$ respectively. We define $\mathcal{F}_{f, T_f}, \mathcal{E}_{e,T_e}$ the subset of $\mathcal{F}_f,\mathcal{E}_e$ with a fixed face tree $T_f$ or edge tree $T_e$ lying over $f$ or $e$ respectively.
Clearly, we have the equalities:
\begin{align*}
|\mathcal{T}| &= |\mathcal{F}| + |\mathcal{E}| \\ 
&= \sum\limits_f |\mathcal{F}_f| + \sum\limits_e |\mathcal{E}_e| \\
&= \sum\limits_{f,T_f} |\mathcal{F}_{f,T_f}| + \sum\limits_{e,T_e} |\mathcal{E}_{e,T_e} |
\end{align*}

We recall the total number of face and edge trees:
\[
N_f = 41873644, \quad N_e = 82
\]

To compute $|\mathcal{F}_{f,T_f}|$, we observe that for a fixed face $f$ on $F$, there are $N_f^{35} \times N_e^{63}$ configurations in $\mathcal{T}$ with a face tree $T_f$ on $f$. In particular, if $T_f$ contains $n$ simplices, then we have 
\[
|\mathcal{F}_{f,T_f}| = \frac{3 (n-1)}{2} \times N_f^{35} \times N_e^{63}
\]
where $3(n-1)/2$ is the number of toric curves added in establishing the face tree.
Thus, we find
\[
|\mathcal{F}| = 36 \times |\mathcal{F}_f| = 36 \times \sum\limits_{T_f} \frac{3(n-1)}{2} \times N_f^{35} \times N_e^{63}
\]
where the quantity 
\[\sum\limits_{T_f} \frac{3(n-1)}{2}=|\mathcal{T}_f|
\]
the total number of curves as in equations~\ref{eq:face}.

To compute $|\mathcal{E}_{e,T_e}|$, we observe that this decomposes into sums over edges lying on an edge of the facet $F$ and edges interior to $F$. Assume that $T_e$ is interior to $F$. Then arguing similarly as above, we find
\[
|\mathcal{E}_{e,T_e}| =  (3(n-1) + 1) \times N_f^{36} \times N_e^{62}
\]
where $3(n-1)+1$ is the number of toric curves added in establishing the edge tree, and $n$ is the number of curve blowups performed.
Similarly, if $T_e$ was on the edge of $F$, then we have
\[
|\mathcal{E}_{e,T_e}| =  2(n-1)+1  \times N_f^{36} \times N_e^{62}
\]
Taking everything together, we find
\begin{align*}
|\mathcal{E}| &= 16 \times \sum\limits_{T_e} (2(n-1)+1)  \times N_f^{36} \times N_e^{62} \\
 &+ 47 \times  \sum\limits_{T_e} (3(n-1) + 1) \times N_f^{36} \times N_e^{62}
\end{align*}
since there are $47 (16)$ choices for $T_e$ that are in the interior (on the edge)  of $F$.

A completely analogous argument demonstrates that the sum $\sum\limits_{t \in \mathcal{T}} P_c(t)$ decomposes as:
\begin{align*}
\sum\limits_{t \in \mathcal{T}} P_c(t) &= \sum\limits_{t \in \mathcal{F}} P_c(t) +  \sum\limits_{t \in \mathcal{E}} P_c(t)\\
&= \sum\limits_{f}\sum\limits_{t \in \mathcal{F}_f} P_c(t) +\sum\limits_{e} \sum\limits_{t \in \mathcal{E}_e} P_c(t) \\
&= \sum\limits_{f,T_f}\sum\limits_{t \in \mathcal{F}_{f,T_f}}P_c(t) +\sum\limits_{e,T_e} \sum\limits_{t \in \mathcal{E}_{e,T_e}} P_c(t)
\end{align*}
For each fixed $f, T_f$ and $e,T_e$, the sums are simply
\begin{align*}
\sum\limits_{t \in \mathcal{F}_{f,T_f}} P_c(t) &= \text{Calib}(T_f)\times N_f^{35} \times N_e^{63} \\
\sum\limits_{t \in \mathcal{E}_{e,T_e}} P_c(t) &= \text{Calib}(T_e)\times N_f^{36} \times N_e^{62} 
\end{align*}
where $\text{Calib}(T)$ is the total number of curves admitting a calibration on a face or edge tree $T$. Summing over all face and edge trees, we find
\begin{align*}
\sum\limits_{T_f}\sum\limits_{t \in \mathcal{F}_{f,T_f}} P_c(t)  &= \sum\limits_{t \in \mathcal{T}_f} P_c(t) \times N_f^{35} \times N_e^{63}\\
\sum\limits_{T_e}\sum\limits_{t \in \mathcal{E}_{e,T_e}} P_c(t) &= \sum\limits_{t \in \mathcal{T}_e} P_c(t) \times N_f^{36} \times N_e^{62}
\end{align*}
where the notation for the first term in the rhs of both expressions are given as in~\ref{eq:face} and \ref{eq:edge}, i.e. is simply the total number of curves admitting a calibration over all face or edge trees. Performing the sum over all $f, e$, we find
\[\sum\limits_{t \in \mathcal{F}} P_c(t) = 36 \times \sum\limits_{t\in T_f}P_c(t) \times N_f^{35} \times N_e^{63}
\]
and
\begin{align*}
\sum\limits_{t \in \mathcal{E}} P_c(t) &= 16 \times (\sum\limits_{t\in T_e}P_c(t) - N_e)\times N_f^{36} \times N_e^{62}\\
&+ 47 \times \sum\limits_{t\in T_e}P_c(t) \times N_f^{36} \times N_e^{62}
\end{align*}

Carrying out the same exercise for curves not admitting a calibration implies a completely analogous expression. Finally, after taking everything together, we find the following:
\begin{align*}
E(\text{Calib}) &= 0.7712\\
E(\text{No Calib}) &= 0.0322.
\end{align*}

\subsection{Example}\label{subsection:example}
We explore a simple example for the purposes of explicitly demonstrating the existence and non-existence of calibrations, and also highlighting the technique used to arrive at the statistics in the above sections. 
Consider the following sequence of blowups on a facet of the polytope of $\mathbb{P}^3$ and denote the last variety with $X$.
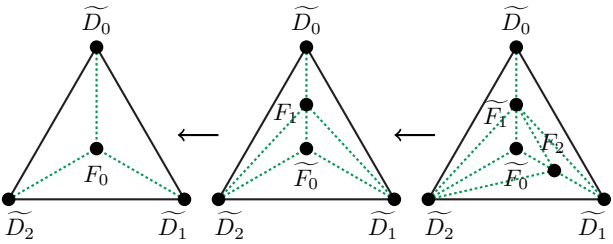
\begin{center}
\begin{tikzpicture}[scale=1.8, every node/.style={scale=.9}]
\draw[thick,color=Black] (90:.75) -- (90+120:.75) -- (90+120+120:.75) -- cycle;
\draw[thick,dash pattern={on 1pt off 1pt},color=ForestGreen] (90:.75) -- (0,0);
\draw[thick,dash pattern={on 1pt off 1pt},color=ForestGreen] (90+120:.75) -- (0,0);
\draw[thick,dash pattern={on 1pt off 1pt},color=ForestGreen] (90+240:.75) -- (0,0);
\fill (90:.75) circle (.5mm);
\fill (90+120:.75) circle (.5mm);
\fill (90+240:.75) circle (.5mm);
\fill (0,0) circle (.5mm);
\node at (90:.95) {$\widetilde{D_0}$}; \node at (90+135:.8) {$\widetilde{D_2}$}; \node at (90+225:.8) {$\widetilde{D_1}$};
\node at (0,-.2) {$F_0$};
\draw[thick,<-] (.6,.1) -- (.9,.1);
\begin{scope}[xshift=1.55cm]
\draw[thick,color=Black] (90:.75) -- (90+120:.75) -- (90+120+120:.75) -- cycle;
\draw[thick,dash pattern={on 1pt off 1pt},color=ForestGreen] (90:.75) -- (0,0);
\draw[thick,dash pattern={on 1pt off 1pt},color=ForestGreen] (90+120:.75) -- (0,0);
\draw[thick,dash pattern={on 1pt off 1pt},color=ForestGreen] (90+240:.75) -- (0,0);
\draw[thick,dash pattern={on 1pt off 1pt},color=ForestGreen] (90+240:.75) -- (90:.325);
\draw[thick,dash pattern={on 1pt off 1pt},color=ForestGreen] (90+120:.75) -- (90:.325);
\fill (90:.75) circle (.5mm);
\fill (90:.325) circle (.5mm);
\fill (90+120:.75) circle (.5mm);
\fill (90+240:.75) circle (.5mm);
\fill (0,0) circle (.5mm);
\node at (120:.3) {$F_1$};
\node at (0,-.2) {$\widetilde{F_0}$};
\node at (90:.95) {$\widetilde{D_0}$}; \node at (90+135:.8) {$\widetilde{D_2}$}; \node at (90+225:.8) {$\widetilde{D_1}$};
\end{scope}
\draw[thick,<-] (2.2,.1) -- (2.5,.1);
\begin{scope}[xshift=3.1cm]
\draw[thick,color=Black] (90:.75) -- (90+120:.75) -- (90+120+120:.75) -- cycle;
\draw[thick,dash pattern={on 1pt off 1pt},color=ForestGreen] (90:.75) -- (0,0);
\draw[thick,dash pattern={on 1pt off 1pt},color=ForestGreen] (90+120:.75) -- (0,0);
\draw[thick,dash pattern={on 1pt off 1pt},color=ForestGreen] (90+240:.75) -- (0,0);
\draw[thick,dash pattern={on 1pt off 1pt},color=ForestGreen] (90+240:.75) -- (90:.325);
\draw[thick,dash pattern={on 1pt off 1pt},color=ForestGreen] (90+120:.75) -- (90:.325);
\draw[thick,dash pattern={on 1pt off 1pt},color=ForestGreen] (90+240:.325) -- (90:.325);
\draw[thick,dash pattern={on 1pt off 1pt},color=ForestGreen] (90+120:.75) -- (90+240:.325);
\fill (90:.75) circle (.5mm);
\fill (90:.325) circle (.5mm);
\fill (90+120:.75) circle (.5mm);
\fill (90+240:.75) circle (.5mm);
\fill (0,0) circle (.5mm);
\fill (90+240:.325) circle (.5mm);
\node at (120:.3) {$\widetilde{F_1}$};
\node at (0,-.2) {$\widetilde{F_0}$};
\node at (8:.27) {$F_2$};
\node at (90:.95) {$\widetilde{D_0}$}; \node at (90+135:.8) {$\widetilde{D_2}$}; \node at (90+225:.8) {$\widetilde{D_1}$};
\end{scope}
\end{tikzpicture}
\end{center}
The relative Mori cone $NE(X/\mathbb{P}^3)$ is generated by the toric curves labeled with green dashes in the above diagram. The intersection matrix between these curves and the divisors in the above diagram are given by the following.
\[\makeatletter\setlength\BA@colsep{5.4pt}\makeatother
\begin{blockarray}{cccccccccc}
 & C_0 & C_1 & C_2 & C_3 & C_4 & C_5& C_6& C_7 & C_8\\
    \begin{block}{c(ccccccccc)}
  \widetilde{D_0} & 0 & 1 & 1 & 2 & 0 & 0 & 0 & 0 & 0 \\
  \widetilde{D_1} & 0 & -1 & 0 & 1 & 0 & 1 & 1 & 3 & 0 \\
  \widetilde{D_2} & 1 & 0 & 0 & 1 & 1 & 0 & 0 & 1 & 1 \\
  \widetilde{F_0} & -3 & 0 & 1 & 0 & -3 & 1 & 1 & 0 & -3 \\
  \widetilde{F_1} & 1  & -1 & -1 & -1 & 1 & 0 & 0 & 1 & 1 \\
   F_2 & 1 & 1 & 0 & 0 & 1 & -1 & -1 & -2 & 1 \\
\end{block}
\end{blockarray}
\]
where the toric curves are given by the following
\begin{align*}
\{C_0, C_1, C_2, C_3, C_4, C_5, C_6, C_7, C_8 \} =\\ \{\widetilde{D_2} \cdot \widetilde{F_0}, \widetilde{D_1} \cdot \widetilde{F_1}, \widetilde{D_2} \cdot \widetilde{F_1}, \widetilde{D_0} \cdot \widetilde{F_1}, \\ \widetilde{F_0} \cdot \widetilde{F_1}, \widetilde{D_2} \cdot F_2, F_2 \cdot \widetilde{F_1}, \widetilde{D_1} \cdot F_2, \widetilde{F_0} \cdot F_2 \}
\end{align*}
The generators of the relative Nef cone $Nef(X/\mathbb{P}^3)$ are then given by dualizing the cone generated by the above rays, and we obtain the following set of generators.
\begin{align*}
&\{J_0, J_1, J_2 \} = \\
&\{-\widetilde{F_0} - 2\widetilde{F_1} - F_2, -\widetilde{F_0} - \widetilde{F_1} - F_2, -2\widetilde{F_0} - 3\widetilde{F_1} -3F_2\}
\end{align*}
For each face leaf $\widetilde{F_a}$, we compute the matrix 
\[Q^a_{ij} = J_i \cdot J_j \cdot \widetilde{F_a}
\]
which yields the following matrices.
\[
Q^0 = \begin{pmatrix}
0 & 0 & 0\\ 
0 & 1 & 0\\
0 & 0 & 0
\end{pmatrix},
\quad
Q^1 = \begin{pmatrix}
2 & 1 & 3\\ 
1 & 0 & 1\\
3 & 1 & 3
\end{pmatrix},
\quad
Q^2 = \begin{pmatrix}
0 & 0 & 0\\ 
0 & 0 & 1\\
0 & 1 & 3
\end{pmatrix}
\]
Thus, we see that the divisors $\widetilde{F_1}, F_2$ can never be calibrated as the matrix $Q^1 - Q^2$ has all non-negative entries. On the other hand, the pairs $\widetilde{F_0}, \widetilde{F_1}$ and $\widetilde{F_0}, F_2$ admit calibrations as the matrices $Q^1 - Q^0$ and $Q^2 - Q^0$ admit both positive and negative entries along their diagonal. For $Q^2 - Q^0$, the equal volume locus is 
\begin{equation}
    J_2 = -\frac{J_1}{3}+\frac{\sqrt{3J_0^2 + 4J_1^2}}{3}
\end{equation}
with $J_1>0$ and $J_2>0$.
For $Q^1-Q^0$ it is
\begin{equation}
J_2 = - \frac{3J_0+J_1}{3}+\frac{\sqrt{2}}{3}\sqrt{3J_0^2 + 2J_1^2}
\end{equation}
if $J_0 > 0$ and $J_1 > (1+\sqrt{2})J_0$.

\noindent {\bf Acknowledgements.}
We thank Fabian Ruehle for discussions. We thank the Northeastern RC team and Keegan Stoner for assistance with the Discovery cluster.  B.S. thanks the 2021 Simons summer workshop at the Simons Center for Geometry and Physics for hospitality during the completion of this work. 
J.H. is supported by NSF CAREER grant PHY-1848089. B.S. is supported by the NSF Graduate Research Fellowship under grant DGE-1451070.

\appendix
\section{Some intersection theory \label{app:intersections}}
Let $N^1(X)$ and $N_1(X)$ be the group of $\mathbb{R}$-divisors and $1$-cycles modulo numerical equivalence respectively. Recall that there exists a perfect, bilinear pairing induced by the intersection product $N^1(X) \times N_1(X) \rightarrow \mathbb{R}$ and hence these vector spaces are canonically dual. Let $\text{Nef}(X) \subset N^1(X)$ and $\overline{\text{NE}}(X) \subset N_1(X)$ be the nef and Mori cones respectively. For simplicity, we will assume that $\text{Nef}(X)$ is finitely generated.

Let $\pi \colon Y \rightarrow X$ be a morphism of projective varieties. A divisor $D \in N^1(Y/X)$ is $\pi$-nef (resp. $\pi$-ample) if $D\cdot C \geq 0$ (resp. $D\cdot C > 0$) for all curves $C$ contracted by $\pi$. We will denote by $NE(Y/X) \subset NE(Y)$ the subcone of curves contracted by $\pi$ and by $Nef(Y/X) \subset N^1(Y/X)$ the cone of $\pi$-nef divisors.

The utility of these definitions is the following basic statement.

\begin{Lemma}[\cite{kollar_mori_1998,54125}]\label{lem:relativeJ}
Assume that $D \in Nef(Y/X)$. Assume in addition either that $-D$ is effective on $Y$ or $D$ is $\pi$-ample. Then there exists an ample class $A \in Nef(X)$ such that $\pi^*A + D \in Nef(Y)$.
\end{Lemma}
\begin{Lemma}\label{lem:relativeJblow}
Assume that $\pi \colon Y \rightarrow X$ is a composition of smooth blowups. Then any divisor $D \in Nef(Y/X)$ satisfies the condition that $-D$ is effective.
\end{Lemma}
\begin{proof}
By \cite[Corollary III.15]{ampcone}, we may assume that $D = a_1\widetilde{E}_1 + \ldots + a_n\widetilde{E}_n$ with $\widetilde{E}_i$ strict transforms of exceptional divisors contained in $Exc(\pi)$. 

It suffices to prove that $a_i \leq 0$ for all $i$. Assume that there exists an $i$ such that $a_i >0$. By \cite[Proposition IV.3]{ampcone}, there exists a sequence of rational curves $C_1, \ldots, C_n \in N_1(Y/X)$ and non-positive integer coefficients $d_{ij}$ such that for all $1 \leq i,j \leq n$, we have
\[
\widetilde{E_i} \cdot (d_{1j} C_1 + \ldots + d_{nj}C_n) = \delta_{ij}
\]
But then we have
\[
D \cdot -(d_{1i}C_1 + \ldots + d_{ni}C_n) = -a_i
\]
which violates the assumption that $D$ is $\pi$-nef since the class $-(d_{1i}C_1 + \ldots + d_{ni}C_n) \in NE(Y/X)$ as all coefficients are non-negative.
\end{proof}

\begin{prop}[\cite{eisenbud_harris_2016}, Theorem 1.23]\label{prop:intfacts}In the following $A^*(X)$ will denote the Chow ring of $X$.
\begin{enumerate}
\item
There exists a uniquely defined intersection product $\cdot \colon A^*(X) \times A^*(X) \rightarrow \mathbb{Z}$ which is trivial on non-complementary dimensions.
\item
The pullback $\pi^* \colon A^c(X) \rightarrow A^c(Y)$ extends to a ring homomorphism on the respective Chow rings.
\item
The pushforward $\pi_* \colon A^*(Y) \rightarrow A^*(X)$ satisfies the following relation:
\[
\pi_*(\pi^*\alpha \cdot \beta) = \alpha \cdot \pi_* \beta \in A_{l-k}(X)
\]
for $\alpha \in A^k(X)$ and $\beta \in A_l(Y)$.
\end{enumerate}
\end{prop}
\begin{corollary}\label{cor:intfacts}
Under the above assumptions, the pullback of nef classes is nef.
\end{corollary}
In the specific construction of interest, we record the following facts of the relevant intersection theory.
\begin{Lemma}\label{lem:intersectionscurve}
Assume in particular that $X,Y$ are threefolds and that $\pi$ is a blowup of a smooth curve $C \subset X$. Then we have the following relations:
\begin{enumerate}
\item
$\pi^* D_1 \cdot \pi^* D_2 \cdot E = 0$ for any divisors $D_1,D_2 \subset X$.
\item
$\pi^* D \cdot E^2 = - D \cdot C$ for any divisor $D \subset X$.
\end{enumerate}
\end{Lemma}
\begin{proof}
The first relation follows from proposition~\ref{prop:intfacts}. The second follows from the fact~\cite{fanovarieties} that the self-intersection class
\[
E^2 = -\pi^*C + \text{deg } \mathcal{N}_{C/X} \cdot f
\]
where $f$ is the fiber class. The result then follows similarly from propostion~\ref{prop:intfacts}.
\end{proof}
\begin{Lemma}\label{lem:intersectionspoint}
Assume that $X, Y$ are threefolds and that $\pi$ is a blowup of a smooth point $p \in X$. Then we have the relations:
\begin{enumerate}
\item
$\pi^* D_1 \cdot \pi^* D_2 \cdot E = 0$ for any divisors $D_1,D_2 \subset X$.
\item
$\pi^* D \cdot E^2 = 0$ for any divisor $D \subset X$.
\end{enumerate}
\end{Lemma}
\begin{proof}
Both relations follow from proposition~\ref{prop:intfacts}.
\end{proof}

\section{General conditions for calibration}
\begin{repLemma}{lem:existence} 
Let $Q \colon \mathbb{R}^n \rightarrow \mathbb{R}$ be a quadratic form. Assume that there exists $a,b \in \mathbb{R}^n_+$ such that $Q(a) > 0$ and $Q(b) < 0$. Then there exists an open neighborhood $U \subset Int(\mathbb{R}^n_+)$ such that $U \cap Q^{-1}(0)$ is a submanifold of codimension $1$.
\end{repLemma}
\begin{proof}
Indeed, as $\mathbb{R}^n_+$ is path-connected and $Q$ is smooth, there must be a point $p \in Int(\mathbb{R}^n_+)$ such that $Q(p) = 0$ by the intermediate value theorem. It suffices to prove that there exists an open subset $U \subset Int(\mathbb{R}^n_+)$ such that $Q|_U^{-1}(0)$ is non-empty and $Q$ is regular on all such points, since then the smooth submersion theorem guarantees that  $Q|_U^{-1}(0)$ is a smooth submanifold of codimension one.

By assumption, $Q$ must be indefinite, and up to a change of basis, we may assume that 
\[Q(x) = a_1x_1^2 + \ldots + a_kx_k^2 - a_{k+1}x_{k+1}^2 - \ldots - a_nx_n^2
\]
with $a_i \geq 0$. If $p$ is a regular point then the Jacobian $\partial Q / \partial x_j$ evaluated at $p$ does not vanish for all $j$, and in particular it must be the case that $x_i(p) \neq 0 $ for some $a_i >0$. Furthermore, we may take $U_p \subset Int(\mathbb{R}^n_+)$ an open ball centered at $p$ sufficiently small such that $x_i(q) > 0$ for all $q \in U_p$\jhc; i.e., $Q$ is regular on $U_p$.

If $p$ is not regular, then the Jacobian vanishes at $p$ and there exists $a_i >0, a_j <0 $, such that $x_i(p) = x_j(p) = 0$. Take $U_p \subset Int(\mathbb{R}^n_+)$ an open ball centered at $p$. Take $\delta_i, \delta_j$ sufficiently small such that the points 
\begin{align*}
p_1 &= (x_1(p), \ldots, x_{i-1}(p), \delta_i, x_{i+1}(p), \ldots, x_n(p))\\
p_2 &= (x_1(p), \ldots, x_{j-1}(p), \delta_j, x_{j+1}(p), \ldots, x_n(p)) 
\end{align*}
are contained in $U_p$.

In particular, $Q(p_1) = a_i\delta_i^2 >0 $ and $Q(p_2) = a_j \delta_j^2 < 0$ and any point $q$ on a straight line between $p_1,p_2$ satisfies $x_i(q), x_j(q) \neq 0$. The first assertion together with the intermediate value theorem implies there exists a distinct point $p' \in U_p$ such that $Q(p') = 0$ and the second assertion implies that this point is regular. Then by the argument of the second paragraph $Q$ is regular on a sufficiently small neighborhood $U_{p'}$.
\end{proof}

\section{Blowups of toric curves}
Our notation in this section will follow the convention introduced in Section~ \ref{curves} regarding proper transforms of exceptional divisors in face trees and edge trees.

We first record a computation which will be used repeatedly in the below. We denote the proper transforms of face leaves to $X_f$ by $\widetilde{F_i}$, their total transforms to $X_e$ by $\pi^*\widetilde{F_i}$ and all pullbacks of edge leaves to $X_e$ by $\pi^*E_i$.
\begin{corollary}\label{cor:intrelations}
We have the following intersection relations:
\begin{enumerate}
\item
$\pi^*E_i \cdot \pi^* \widetilde{F_j} \cdot \pi^*\widetilde{F_k} = 0$.
\item
$\pi^*E_i\cdot \pi^*E_j \cdot \pi^* \widetilde{F_k} = 0$ for $ i \neq j$.
\item
$(\pi^*E_i)^2 \cdot \pi^*\widetilde{F_j} = -C \cdot \widetilde{F_j}$ on $X_f$.
\end{enumerate}
\end{corollary}
\begin{proof}
The first two relations follow from Lemma~\ref{lem:intersectionscurve}(1). To see the third relation, fix $E_i$ and consider the truncation of the sequence of blowups $\pi_i \colon X_i \rightarrow \ldots \rightarrow X_f$ establishing an edge tree. Moreover, we may assume that each step $X_{j+1}\rightarrow X_{j}$ is the blowup along one of two toric curves contained in the exceptional divisor $E_{j}$ of the map $X_{j} \rightarrow X_{j-1}$.

By Lemma~\ref{lem:intersectionscurve}(2), we have 
\begin{align*}
&(\pi^*E_i)^2 \cdot \pi^*\widetilde{F_j} = \pi^*E_i^2 \cdot \pi^*\widetilde{F}_j \\
= &-\pi^* (\widetilde{E_a} \cdot \widetilde{E_b} )\cdot \pi^* \widetilde{F}_j +  \text{deg } \mathcal{N}_{E_a \cap E_b/X_{i-1}} f \cdot \pi^*\widetilde{F}_j
\end{align*}
where $\widetilde{E_a}, \widetilde{E_b}$ are the toric divisors such that $X_{i} \rightarrow X_{i-1}$ is the blowup along $E_a \cap E_b$ and $f$ is the fiber of the projective bundle $E_i \rightarrow E_a \cap E_b$. By proposition~\ref{prop:intfacts}(3), we have $f \cdot \pi^*\widetilde{F}_j = \pi_*f \cdot \widetilde{F_j}= 0$. By the assumption of the last sentence in the above paragraph, we may assume that $\widetilde{E_a} = E_{i-1}$ is the exceptional divisor of the blowup $X_{i-1} \rightarrow X_{i-2}$ and $\widetilde{E_b} = \pi^*E_b - \ldots - E_a$ as classes on $X_{i-1}$.

Thus, we have the equalities
\begin{align*}
&-\pi^* (\widetilde{E_a} \cdot \widetilde{E_b} )\cdot \pi^* \widetilde{F}_j \\ =
&-\pi^*E_{i-1} \cdot (\pi^*E_b - \ldots - \pi^*E_{i-1})\cdot \pi^*\widetilde{F}_j  \\
= &\ \pi^* E_{i-1}^2 \cdot \pi^* \widetilde{F_j}\\
        &\qquad\quad\vdotswithin{ = }\notag \\
= &\ \pi^*E_{f+1}^2 \cdot \pi^*\widetilde{F}_j\\
= &-C \cdot \widetilde{F_j}
\end{align*}
where the second equality follows again by Lemma~\ref{lem:intersectionscurve}(1), the second to last equality follows by induction on the edge tree and $C$ is the base curve of the edge tree.
\end{proof}

\begin{corollary}\label{cor:ampconditions}
Let $D_x,D_y \in N^1(X_f)$ denote the unique toric divisors adjacent to the base curve of the edge tree, $\widetilde{D_x},\widetilde{D_y}$ their proper transforms on $X_e$. Let
\[J_{rel} = -a_{f+1} \pi^*E_{f+1} - \ldots -a_e E_{e} \in N^1(X_e/X_f)
\]
 with $a_l$ arbitrary coefficients. Let $\widetilde{E_i}$ denote the corresponding proper transforms on $X_e$. Then:
\begin{enumerate}
\item
$\widetilde{E_i} = \pi^*E_{i} - \pi^*E_{i_1} - \ldots - \pi^*E_{i_k}$ for all $i$ for some $k$.
\item
$\widetilde{D_x} \cdot \widetilde{E_i} \cdot J_{rel} = a_i - a_{i_1} - \ldots - a_{i_k}$.
\end{enumerate}
\end{corollary}
\begin{proof}
Part $1$ follows from the standard identification of the strict transform with the total transform.

From the fact that $\widetilde{D_x} = \pi^*D_x$ and part 1, we have the following:
\begin{align*}
\widetilde{D_x} \cdot \widetilde{E_i} \cdot J_{rel} = &\,\pi^*D_x \cdot (\pi^*E_{i} - \pi^*E_{i_1} - \ldots - \pi^*E_{i_k}) \\ &\cdot (-a_{f+1} \pi^*E_{f+1} - \ldots -a_e E_{e}) \\
\end{align*}
By corollary~\ref{cor:intrelations}(2), all terms in the above expression are trivial except for the terms $\pi^*D_x \cdot \pi^*E_j^2$. We then have the equalities
\[
\pi^*D_x \cdot \pi^*E_j^2 = -C \cdot D_x = -1
\]
where the first equality follows from corollary~\ref{cor:intrelations}(3), $C$ is the base curve of the edge tree, and the second equality follows from the assumption that $D_x$ is adjacent to the base curve of the edge tree. The conclusion then follows.
\end{proof}
\begin{repLemma}{lem:redcurves}
The proper transforms $\widetilde{F}_j$ and $\widetilde{E}_k$ can always be strongly calibrated on $X_e$.
\end{repLemma}
\begin{proof}
Consider the truncation to the chain of blowups $\pi: X_k \rightarrow \ldots \rightarrow X_f \rightarrow \ldots \rightarrow X_{j-1}$.

It suffices to prove by Lemma~\ref{lem:calibexc} that there exists $J \in Nef(X_k)$ such that $J^2 \cdot (E_k - \widetilde{F_j}) >0$ where the intersections are defined on the variety $X_k$. Indeed, any relative ample class must be of the form
\[
J_{rel} = - a_{j} \pi^* F_{j}  - \ldots - a_{f} \pi^*F_f - \ldots - a_k E_k
\]
for arbitrary coefficients $a_l$. Moreover, we may assume $a_l >0$ for $ l \geq f+1$ by applying corollary~\ref{cor:ampconditions}(2). Then by Lemma~\ref{lem:relativeJ}, there exists $J \in Nef(X_k)$ of the form:
\[J = (\pi^*H - a_{j} \pi^* F_{j}  - \ldots - a_{f} \pi^*F_f - \ldots - a_k E_k)
\] 
where $a_k > 0$ and $H$ is ample on $X_{j-1}$.

Let $C$ denote the toric curve on $X$ corresponding to the base of the edge tree. We claim that there exists $J_0 \in Nef(X_f)$ such that:
\begin{enumerate}[(a)]
\item
$J_0 \cdot \widetilde{C} > 0 $ where $\widetilde{C}$ is the proper transform of $C$ under the morphism $f \colon X_f \rightarrow X$.
\item
$J_0 \cdot V = 0$ for any curve $V$ contained in the exceptional locus $Exc(f)$
\end{enumerate}
Indeed, let $G $ be any ample class on $X$. Then clearly, $f^*G \cdot V = G \cdot f_* V = 0$ for $V \subset Exc(f)$ by proposition~\ref{prop:intfacts}(3). By assumption $C = D_1 \cap D_2$ is the intersection of toric divisors on $X$ and in particular, the proper transform is given by the class $\widetilde{C} = \widetilde{D_1} \cdot \widetilde{D_2}$ where
\begin{align*}
\widetilde{D_1}&= f^*D_1 -f^* F_{a_1} - \ldots - f^* F_{a_m}\\
\widetilde{D_2} &= f^* D_2 - f^*F_{b_1} - \ldots - f^* F_{b_n}
\end{align*}
By assumption that $f$ is a morphism establishing a face tree, we may assume that any exceptional divisor associated to a curve blowup at an intermediate stage of $f$ (i.e., a curve blow-up inside the face) appears in at most one of $\widetilde{D_1}$ and $\widetilde{D_2}$. Thus, we have
\begin{align*}
f^* G \cdot \widetilde{D_1}\cdot \widetilde{D_2} &= f^*G \cdot (f^*D_1 - \ldots - f^*F_{a_m}) \\
&\cdot (f^*D_2 - \ldots - f^*F_{b_n})\\
&= f^* G \cdot f^*D_1 \cdot f^*D_2 \\
&= G \cdot D_1 \cdot D_2 > 0
\end{align*}
where the second equality follows from the fact that $f^*G \cdot f^*F_{a_k} \cdot f^*F_{a_l} = 0 $ as $f^*F_{a_k} \cdot f^*F_{a_l} \subset Exc(f)$,
and so the triple intersection is trivial in the Chow ring by \ref{prop:intfacts}(3), since $f$ contracts $F_{a_k}$ and $F_{a_l}$ to a point. This exstablishes the claim for any ample class $G$.

Consider the case where $f^*G = J_0$. As $J_0$ is nef and the nef cone is closed under sums, we have that
\[J_n \coloneqq n\pi^*J_0 + \pi^*H - a_{j} \pi^* F_{j}  - \ldots - a_k E_k
\]
is nef for any $n > 1$. By Lemma~\ref{lem:intersectionscurve}(1), we have that $\pi^*J_0 \cdot \pi^*J_0 \cdot E_k$, $\pi^*J_0 \cdot \pi^* H \cdot E_k$, $\pi^*J_0 \cdot \pi^* F_a \cdot E_k$, $\pi^*J_0 \cdot \pi^* E_b \cdot E_k$, are all $0$ unless $k = b$ in the last case, where we have 
\[\pi^* J_0 \cdot E_k^2 = \pi^* J_0 \cdot (-\pi^* C) = - J_0\cdot C\]
by Corollary~\ref{cor:intrelations} and hence 
\begin{align*}
&J_n^2 \cdot E_k\\ &= (n\pi^* J_0)(n\pi^*J_0 - 2(\pi^*H - a_{j} \pi^* F_{j}  +\ldots + a_k E_k))\cdot E_k\\ &+ (\pi^*H - a_{j} \pi^* F_{j}  -\ldots - a_k E_k)^2 \cdot E_k \\
&= 2na_kJ_0\cdot C + (\pi^*H - a_{j} \pi^* F_{j}  -\ldots - a_k E_k)^2 \cdot E_k
\end{align*}

Now, we may assume $\widetilde{F_j} = \pi^* F_j - \pi^* F_a -  \ldots - \pi^* F_b$ for some $j < a < b < i$. Then again by Lemma~\ref{lem:intersectionscurve}(1) and claim (b) in the above, we have that 
\[
J_n^2 \cdot \widetilde{F_j} = (\pi^*H -a_j \pi^* F_0 - \ldots - a_kE_k)^2 \cdot \widetilde{F_j}.
\]
Taking everything together, we have
\begin{align*}
J_n^2\cdot (E_k - \widetilde{F_j}) &= 2 n a_k J_0\cdot C\\
&+ (\pi^*H -a_j \pi^* F_j - \ldots - a_iE_i)^2 \cdot (E_k - \widetilde{F_j})
\end{align*}
Indeed, as the second term is independent of $n$, we find that $J_n^2 \cdot (E_k - \widetilde{F_j})$ is positive for sufficiently large $n$ since $J_0$ is ample and $a_k > 0$.

\end{proof}
\begin{repLemma}{lem:edgecurves}
Assume that a pair of proper transforms $\widetilde{E_j}, \widetilde{E_k}$ intersect. Then the pair admits a calibration.
\end{repLemma}
\begin{proof}
We follow the proof of Lemma~\ref{lem:redcurves}. We consider the truncation of the chain of blowups $\pi \colon X_k \rightarrow \ldots X_j \rightarrow X_{j-1}$ which we may assume is a sequence of only edge blowups on a single curve. 

As in the above, it suffices to prove that there exists $J \in Nef(X_k)$ such that $J^2 \cdot (E_k - \widetilde{E_j}) > 0$. In this setting, we may assume that the proper transform of $E_j$ takes the form
\[
\widetilde{E_j} = \pi^* E_j - \pi^* E_a -  \ldots - \pi^* E_b
\]
for some $j < a \leq b \leq k$. Arguing similarly as in the proof of Lemma~\ref{lem:redcurves}, we consider the nef class
\[J_n \coloneqq  n\pi^*H - a_{j} \pi^* E_{j}  - \ldots - a_k E_k
\]
where $H$ is an ample divisor on $X_{j-1}$. This yields
\begin{align*}
J_n^2 \cdot (E_k - \widetilde{E_j}) &= 2n(a_k - (a_j - a_a - \ldots - a_b)) H \cdot C\\
&+ (- a_j \pi^* E_j - \ldots - a_kE_k)^2 \cdot (E_k - \widetilde{E_j})
\end{align*}
It suffices to prove that there exists coefficients $a_i$ satisfying $a_k - (a_j - a_a - \ldots - a_b) > 0$ such that the sum $D = -a_j \pi^*E_j - \ldots - a_kE_k \in Nef(X_k/X_{j-1})$ is relatively nef. We then conclude by Lemma~\ref{lem:relativeJ}. 

By the toric Mori theorem, the cone $NE(X_k/X_{j-1})$ is generated by the toric curves $\pi^*D_x \cdot \widetilde{E_i}$ and $\pi^*D_y \cdot \widetilde{E_i}$ for $j \leq i \leq k$ where $D_x,D_y$ are the unique toric divisors on $X_f$ adjacent to the base curve of the edge tree. In particular, the divisor $D \in Nef(X_k/X_{j-1})$ if and only if $D \cdot \pi^*D_x \cdot \widetilde{E_i},D \cdot \pi^*D_y \cdot \widetilde{E_i} \geq 0$ for all $j \leq i \leq k$ on $X_k$.

By corollary~\ref{cor:ampconditions}(1), each proper transform $\widetilde{E_i}$ takes the form
\[
\widetilde{E_i}= \pi^*E_i - \pi^* E_a - \ldots - \pi^* E_b
\]
for $i < a \leq b \leq k$. Applying corollary~\ref{cor:ampconditions}(2), the above inequalities reduces to the following system:
\begin{align*}
a_k &\geq 0 \\
        & \vdotswithin{ = }\notag \\
a_j - a_a - \ldots - a_b &\geq 0 
\end{align*}
Choosing $a_k >0$ and setting all other inequalities to equalities, this yields an upper triangular matrix with nonzero entries on the diagonal. Thus, it is clear that the desired condition can be satisfied.
\end{proof}
\begin{repLemma}{lem:nonexistence}
Assume $X_n \rightarrow \ldots \rightarrow X_f$ a sequence of blowups establishing a collection of edge trees. Assume that $J^2 \cdot (D_1 - D_2) > 0$ for any ample class $J \in Nef(X_f)$ with $D_1,D_2$ toric divisors on the same face tree; that is, $D_1$ and $D_2$ cannot be calibrated. Assume in addition that $D_1$ does not intersect with any base curve (in the original variety $X$) of an edge tree. Then any ample class $J' \in Nef(X_n)$ satisfies $J'^2 \cdot (\widetilde{D_1} - \widetilde{D_2}) > 0$ on $X_n$; that is, the proper transforms $\widetilde{D_1}$ and $\widetilde{D_2}$ cannot be calibrated. In other words, the non-existence of a calibration persists under such blowups.
\end{repLemma}
\begin{proof}
We may reduce to the case with the sequence
\[
\pi \colon X_n \rightarrow \ldots \rightarrow X_f \rightarrow \ldots \rightarrow X
\]
where $X_n \rightarrow \ldots \rightarrow X_f$ is a collection of blowups establishing edge trees on the three adjacent toric curves of the face tree on $X_f$.

Let $J' \in Nef(X_n)$ be any ample class. Then $J' = \pi^*D + \pi^*F + \pi^*E$ where $D \in N^1(X)$, $F$ a sum of exceptional divisors contained in the face tree, and $E$ a sum of exceptional divisors contained in the edge trees. Let $C \subset X_m$ be any toric curve contained in the face tree. Then $\pi^*C$ is effective and we must have
\[
J' \cdot \pi^* C = (\pi^*D + \pi^*F + \pi^*E) \cdot \pi^*C > 0
\]
By proposition~\ref{prop:intfacts}(3), we have $\pi^*D \cdot \pi^*C = D \cdot p = 0$ as $\pi^*C \in Exc(\pi)$ and $p$ is a point. By Lemma~\ref{lem:intersectionscurve}(1), $\pi^*E \cdot \pi^*C = 0$. So, it must be the case that $F$ satisfies $F \cdot C = \pi^*F \cdot \pi^* C  > 0$. 

In particular, we have $F\in Amp(X_f/X)$ and by Lemma~\ref{lem:relativeJ}, $\pi^*A + F$ is nef on $X_f$ for some $A \in X$ ample. By assumption 
\[(\pi^*A +F)^2 \cdot (D_1 - D_2) = F^2 \cdot (D_1 - D_2) > 0
\] where $\pi^*A \cdot \pi^*A \cdot D_i = \pi^*A \cdot F \cdot D_i = 0$ by Proposition~\ref{prop:intfacts}(3) and the fact that $D_i, F\cdot D_i \subset Exc(f)$.

But for any divisors $D_1, D_2$ contained in the face tree, we have
\begin{align*}
J'^2 \cdot (\widetilde{D_1} - \widetilde{D_2}) &= (\pi^*D + \pi^* F + \pi^*E)^2 \cdot (\pi^*D_1 - \pi^*D_2)\\
&= F^2 \cdot (D_1 - D_2) + E^2 \cdot (\pi^*D_1 - \pi^*D_2)
\end{align*}
where all the other terms vanish by applying proposition~\ref{prop:intfacts}(c). By corollary~\ref{cor:intrelations}(2), the second term takes the form
\begin{align*}
&\,E^2 \cdot (\pi^*D_1 - \pi^*D_2)\\ &= (a_f \pi^*E_f + \ldots + a_n \pi^*E_n)^2 \cdot (\pi^*D_1 - \pi^*D_2) \\
&= (a_f^2 \pi^*E_f^2 + \ldots + a_n^2 \pi^*E_n^2) \cdot (\pi^*D_1 - \pi^*D_2)
\end{align*}
By assumption $D_1$ does not intersect the base curve of any edge tree, so by corollary~\ref{cor:intrelations}(3), we have 
\[(a_f^2 \pi^*E_f^2 + \ldots + a_n^2 \pi^*E_n^2)\cdot \pi^*D_1 = 0
\]
On the other hand, again by corollary~\ref{cor:intrelations}(3), 
\[-a_j^2 \pi^*E_f^2 \cdot \pi^*D_2 = a_j^2 C \cdot D_2 \geq 0
\]
on $X_f$ which is nonzero if and only if $D_2$ is adjacent to a base curve of one of the edge trees.
\end{proof}

\bibliography{refs}

\begin{thebibliography}{54}%
\makeatletter
\providecommand \@ifxundefined [1]{%
 \@ifx{#1\undefined}
}%
\providecommand \@ifnum [1]{%
 \ifnum #1\expandafter \@firstoftwo
 \else \expandafter \@secondoftwo
 \fi
}%
\providecommand \@ifx [1]{%
 \ifx #1\expandafter \@firstoftwo
 \else \expandafter \@secondoftwo
 \fi
}%
\providecommand \natexlab [1]{#1}%
\providecommand \enquote  [1]{``#1''}%
\providecommand \bibnamefont  [1]{#1}%
\providecommand \bibfnamefont [1]{#1}%
\providecommand \citenamefont [1]{#1}%
\providecommand \href@noop [0]{\@secondoftwo}%
\providecommand \href [0]{\begingroup \@sanitize@url \@href}%
\providecommand \@href[1]{\@@startlink{#1}\@@href}%
\providecommand \@@href[1]{\endgroup#1\@@endlink}%
\providecommand \@sanitize@url [0]{\catcode `\\12\catcode `\$12\catcode
  `\&12\catcode `\#12\catcode `\^12\catcode `\_12\catcode `\%12\relax}%
\providecommand \@@startlink[1]{}%
\providecommand \@@endlink[0]{}%
\providecommand \url  [0]{\begingroup\@sanitize@url \@url }%
\providecommand \@url [1]{\endgroup\@href {#1}{\urlprefix }}%
\providecommand \urlprefix  [0]{URL }%
\providecommand \Eprint [0]{\href }%
\providecommand \doibase [0]{http://dx.doi.org/}%
\providecommand \selectlanguage [0]{\@gobble}%
\providecommand \bibinfo  [0]{\@secondoftwo}%
\providecommand \bibfield  [0]{\@secondoftwo}%
\providecommand \translation [1]{[#1]}%
\providecommand \BibitemOpen [0]{}%
\providecommand \bibitemStop [0]{}%
\providecommand \bibitemNoStop [0]{.\EOS\space}%
\providecommand \EOS [0]{\spacefactor3000\relax}%
\providecommand \BibitemShut  [1]{\csname bibitem#1\endcsname}%
\let\auto@bib@innerbib\@empty
\bibitem [{\citenamefont {Halverson}\ \emph
  {et~al.}(2017{\natexlab{a}})\citenamefont {Halverson}, \citenamefont {Long},\
  and\ \citenamefont {Sung}}]{Halverson:2017ffz}%
  \BibitemOpen
  \bibfield  {author} {\bibinfo {author} {\bibfnamefont {J.}~\bibnamefont
  {Halverson}}, \bibinfo {author} {\bibfnamefont {C.}~\bibnamefont {Long}}, \
  and\ \bibinfo {author} {\bibfnamefont {B.}~\bibnamefont {Sung}},\ }\href
  {\doibase 10.1103/PhysRevD.96.126006} {\bibfield  {journal} {\bibinfo
  {journal} {Phys. Rev. D}\ }\textbf {\bibinfo {volume} {96}},\ \bibinfo
  {pages} {126006} (\bibinfo {year} {2017}{\natexlab{a}})},\ \Eprint
  {http://arxiv.org/abs/1706.02299} {arXiv:1706.02299 [hep-th]} \BibitemShut
  {NoStop}%
\bibitem [{\citenamefont {Georgi}\ and\ \citenamefont
  {Glashow}(1974)}]{georgiglashow}%
  \BibitemOpen
  \bibfield  {author} {\bibinfo {author} {\bibfnamefont {H.}~\bibnamefont
  {Georgi}}\ and\ \bibinfo {author} {\bibfnamefont {S.~L.}\ \bibnamefont
  {Glashow}},\ }\href {\doibase 10.1103/PhysRevLett.32.438} {\bibfield
  {journal} {\bibinfo  {journal} {Phys. Rev. Lett.}\ }\textbf {\bibinfo
  {volume} {32}},\ \bibinfo {pages} {438} (\bibinfo {year} {1974})}\BibitemShut
  {NoStop}%
\bibitem [{\citenamefont {Pati}\ and\ \citenamefont {Salam}(1974)}]{patisalam}%
  \BibitemOpen
  \bibfield  {author} {\bibinfo {author} {\bibfnamefont {J.~C.}\ \bibnamefont
  {Pati}}\ and\ \bibinfo {author} {\bibfnamefont {A.}~\bibnamefont {Salam}},\
  }\href {\doibase 10.1103/PhysRevD.10.275} {\bibfield  {journal} {\bibinfo
  {journal} {Phys. Rev. D}\ }\textbf {\bibinfo {volume} {10}},\ \bibinfo
  {pages} {275} (\bibinfo {year} {1974})}\BibitemShut {NoStop}%
\bibitem [{\citenamefont {Langacker}(1981)}]{langackerGUT}%
  \BibitemOpen
  \bibfield  {author} {\bibinfo {author} {\bibfnamefont {P.}~\bibnamefont
  {Langacker}},\ }\href {\doibase https://doi.org/10.1016/0370-1573(81)90059-4}
  {\bibfield  {journal} {\bibinfo  {journal} {Physics Reports}\ }\textbf
  {\bibinfo {volume} {72}},\ \bibinfo {pages} {185} (\bibinfo {year}
  {1981})}\BibitemShut {NoStop}%
\bibitem [{\citenamefont {Geng}\ and\ \citenamefont
  {Marshak}(1989)}]{PhysRevD.39.693}%
  \BibitemOpen
  \bibfield  {author} {\bibinfo {author} {\bibfnamefont {C.~Q.}\ \bibnamefont
  {Geng}}\ and\ \bibinfo {author} {\bibfnamefont {R.~E.}\ \bibnamefont
  {Marshak}},\ }\href {\doibase 10.1103/PhysRevD.39.693} {\bibfield  {journal}
  {\bibinfo  {journal} {Phys. Rev. D}\ }\textbf {\bibinfo {volume} {39}},\
  \bibinfo {pages} {693} (\bibinfo {year} {1989})}\BibitemShut {NoStop}%
\bibitem [{\citenamefont {Minahan}\ \emph {et~al.}(1990)\citenamefont
  {Minahan}, \citenamefont {Ramond},\ and\ \citenamefont
  {Warner}}]{PhysRevD.41.715}%
  \BibitemOpen
  \bibfield  {author} {\bibinfo {author} {\bibfnamefont {J.~A.}\ \bibnamefont
  {Minahan}}, \bibinfo {author} {\bibfnamefont {P.}~\bibnamefont {Ramond}}, \
  and\ \bibinfo {author} {\bibfnamefont {R.~C.}\ \bibnamefont {Warner}},\
  }\href {\doibase 10.1103/PhysRevD.41.715} {\bibfield  {journal} {\bibinfo
  {journal} {Phys. Rev. D}\ }\textbf {\bibinfo {volume} {41}},\ \bibinfo
  {pages} {715} (\bibinfo {year} {1990})}\BibitemShut {NoStop}%
\bibitem [{\citenamefont {Geng}\ and\ \citenamefont
  {Marshak}(1990)}]{PhysRevD.41.717}%
  \BibitemOpen
  \bibfield  {author} {\bibinfo {author} {\bibfnamefont {C.~Q.}\ \bibnamefont
  {Geng}}\ and\ \bibinfo {author} {\bibfnamefont {R.~E.}\ \bibnamefont
  {Marshak}},\ }\href {\doibase 10.1103/PhysRevD.41.717} {\bibfield  {journal}
  {\bibinfo  {journal} {Phys. Rev. D}\ }\textbf {\bibinfo {volume} {41}},\
  \bibinfo {pages} {717} (\bibinfo {year} {1990})}\BibitemShut {NoStop}%
\bibitem [{\citenamefont {Langacker}(1993)}]{Langacker:1993ai}%
  \BibitemOpen
  \bibfield  {author} {\bibinfo {author} {\bibfnamefont {P.}~\bibnamefont
  {Langacker}},\ }in\ \href@noop {} {\emph {\bibinfo {booktitle}
  {{International Symposium on Neutrino Astrophysics}}}}\ (\bibinfo {year}
  {1993})\ \Eprint {http://arxiv.org/abs/hep-ph/9303304} {arXiv:hep-ph/9303304}
  \BibitemShut {NoStop}%
\bibitem [{\citenamefont {Ellis}\ \emph {et~al.}(1991)\citenamefont {Ellis},
  \citenamefont {Kelley},\ and\ \citenamefont {Nanopoulos}}]{Ellis:1990wk}%
  \BibitemOpen
  \bibfield  {author} {\bibinfo {author} {\bibfnamefont {J.~R.}\ \bibnamefont
  {Ellis}}, \bibinfo {author} {\bibfnamefont {S.}~\bibnamefont {Kelley}}, \
  and\ \bibinfo {author} {\bibfnamefont {D.~V.}\ \bibnamefont {Nanopoulos}},\
  }\href {\doibase 10.1016/0370-2693(91)90980-5} {\bibfield  {journal}
  {\bibinfo  {journal} {Phys. Lett. B}\ }\textbf {\bibinfo {volume} {260}},\
  \bibinfo {pages} {131} (\bibinfo {year} {1991})}\BibitemShut {NoStop}%
\bibitem [{\citenamefont {Langacker}\ and\ \citenamefont
  {Luo}(1991)}]{Langacker:1991an}%
  \BibitemOpen
  \bibfield  {author} {\bibinfo {author} {\bibfnamefont {P.}~\bibnamefont
  {Langacker}}\ and\ \bibinfo {author} {\bibfnamefont {M.-x.}\ \bibnamefont
  {Luo}},\ }\href {\doibase 10.1103/PhysRevD.44.817} {\bibfield  {journal}
  {\bibinfo  {journal} {Phys. Rev. D}\ }\textbf {\bibinfo {volume} {44}},\
  \bibinfo {pages} {817} (\bibinfo {year} {1991})}\BibitemShut {NoStop}%
\bibitem [{\citenamefont {Giunti}\ \emph {et~al.}(1991)\citenamefont {Giunti},
  \citenamefont {Kim},\ and\ \citenamefont {Lee}}]{Giunti:1991ta}%
  \BibitemOpen
  \bibfield  {author} {\bibinfo {author} {\bibfnamefont {C.}~\bibnamefont
  {Giunti}}, \bibinfo {author} {\bibfnamefont {C.~W.}\ \bibnamefont {Kim}}, \
  and\ \bibinfo {author} {\bibfnamefont {U.~W.}\ \bibnamefont {Lee}},\ }\href
  {\doibase 10.1142/S0217732391001883} {\bibfield  {journal} {\bibinfo
  {journal} {Mod. Phys. Lett. A}\ }\textbf {\bibinfo {volume} {6}},\ \bibinfo
  {pages} {1745} (\bibinfo {year} {1991})}\BibitemShut {NoStop}%
\bibitem [{\citenamefont {Weinberg}(1976)}]{Weinberg:1975gm}%
  \BibitemOpen
  \bibfield  {author} {\bibinfo {author} {\bibfnamefont {S.}~\bibnamefont
  {Weinberg}},\ }\href {\doibase 10.1103/PhysRevD.19.1277} {\bibfield
  {journal} {\bibinfo  {journal} {Phys. Rev. D}\ }\textbf {\bibinfo {volume}
  {13}},\ \bibinfo {pages} {974} (\bibinfo {year} {1976})},\ \bibinfo {note}
  {[Addendum: Phys.Rev.D 19, 1277--1280 (1979)]}\BibitemShut {NoStop}%
\bibitem [{\citenamefont {Weinberg}(1989)}]{Weinberg:1988cp}%
  \BibitemOpen
  \bibfield  {author} {\bibinfo {author} {\bibfnamefont {S.}~\bibnamefont
  {Weinberg}},\ }\href {\doibase 10.1103/RevModPhys.61.1} {\bibfield  {journal}
  {\bibinfo  {journal} {Rev. Mod. Phys.}\ }\textbf {\bibinfo {volume} {61}},\
  \bibinfo {pages} {1} (\bibinfo {year} {1989})}\BibitemShut {NoStop}%
\bibitem [{\citenamefont {Cveti{\v c}}\ \emph {et~al.}(2019)\citenamefont
  {Cveti{\v c}}, \citenamefont {Halverson}, \citenamefont {Lin}, \citenamefont
  {Liu},\ and\ \citenamefont {Tian}}]{Cvetic:2019gnh}%
  \BibitemOpen
  \bibfield  {author} {\bibinfo {author} {\bibfnamefont {M.}~\bibnamefont
  {Cveti{\v c}}}, \bibinfo {author} {\bibfnamefont {J.}~\bibnamefont
  {Halverson}}, \bibinfo {author} {\bibfnamefont {L.}~\bibnamefont {Lin}},
  \bibinfo {author} {\bibfnamefont {M.}~\bibnamefont {Liu}}, \ and\ \bibinfo
  {author} {\bibfnamefont {J.}~\bibnamefont {Tian}},\ }\href {\doibase
  10.1103/PhysRevLett.123.101601} {\bibfield  {journal} {\bibinfo  {journal}
  {Phys. Rev. Lett.}\ }\textbf {\bibinfo {volume} {123}},\ \bibinfo {pages}
  {101601} (\bibinfo {year} {2019})},\ \Eprint
  {http://arxiv.org/abs/1903.00009} {arXiv:1903.00009 [hep-th]} \BibitemShut
  {NoStop}%
\bibitem [{\citenamefont {Cveti}\ \emph {et~al.}(2020)\citenamefont {Cveti},
  \citenamefont {Halverson}, \citenamefont {Lin},\ and\ \citenamefont
  {Long}}]{Cvetic:2020fkd}%
  \BibitemOpen
  \bibfield  {author} {\bibinfo {author} {\bibfnamefont {M.}~\bibnamefont
  {Cveti}}, \bibinfo {author} {\bibfnamefont {J.}~\bibnamefont {Halverson}},
  \bibinfo {author} {\bibfnamefont {L.}~\bibnamefont {Lin}}, \ and\ \bibinfo
  {author} {\bibfnamefont {C.}~\bibnamefont {Long}},\ }\href {\doibase
  10.1103/PhysRevD.102.026012} {\bibfield  {journal} {\bibinfo  {journal}
  {Phys. Rev. D}\ }\textbf {\bibinfo {volume} {102}},\ \bibinfo {pages}
  {026012} (\bibinfo {year} {2020})},\ \Eprint
  {http://arxiv.org/abs/2004.00630} {arXiv:2004.00630 [hep-th]} \BibitemShut
  {NoStop}%
\bibitem [{\citenamefont {Gato-Rivera}\ and\ \citenamefont
  {Schellekens}(2014)}]{Gato-Rivera:2014afa}%
  \BibitemOpen
  \bibfield  {author} {\bibinfo {author} {\bibfnamefont {B.}~\bibnamefont
  {Gato-Rivera}}\ and\ \bibinfo {author} {\bibfnamefont {A.~N.}\ \bibnamefont
  {Schellekens}},\ }\href {\doibase 10.1016/j.nuclphysb.2014.03.026} {\bibfield
   {journal} {\bibinfo  {journal} {Nucl. Phys. B}\ }\textbf {\bibinfo {volume}
  {883}},\ \bibinfo {pages} {529} (\bibinfo {year} {2014})},\ \Eprint
  {http://arxiv.org/abs/1401.1782} {arXiv:1401.1782 [hep-ph]} \BibitemShut
  {NoStop}%
\bibitem [{\citenamefont {Dienes}\ and\ \citenamefont
  {Faraggi}(1995)}]{Dienes:1995bx}%
  \BibitemOpen
  \bibfield  {author} {\bibinfo {author} {\bibfnamefont {K.~R.}\ \bibnamefont
  {Dienes}}\ and\ \bibinfo {author} {\bibfnamefont {A.~E.}\ \bibnamefont
  {Faraggi}},\ }\href {\doibase 10.1016/0550-3213(95)00497-1} {\bibfield
  {journal} {\bibinfo  {journal} {Nucl. Phys. B}\ }\textbf {\bibinfo {volume}
  {457}},\ \bibinfo {pages} {409} (\bibinfo {year} {1995})},\ \Eprint
  {http://arxiv.org/abs/hep-th/9505046} {arXiv:hep-th/9505046} \BibitemShut
  {NoStop}%
\bibitem [{\citenamefont {Dijkstra}\ \emph {et~al.}(2005)\citenamefont
  {Dijkstra}, \citenamefont {Huiszoon},\ and\ \citenamefont
  {Schellekens}}]{Dijkstra:2004cc}%
  \BibitemOpen
  \bibfield  {author} {\bibinfo {author} {\bibfnamefont {T.~P.~T.}\
  \bibnamefont {Dijkstra}}, \bibinfo {author} {\bibfnamefont {L.~R.}\
  \bibnamefont {Huiszoon}}, \ and\ \bibinfo {author} {\bibfnamefont {A.~N.}\
  \bibnamefont {Schellekens}},\ }\href {\doibase
  10.1016/j.nuclphysb.2004.12.032} {\bibfield  {journal} {\bibinfo  {journal}
  {Nucl. Phys. B}\ }\textbf {\bibinfo {volume} {710}},\ \bibinfo {pages} {3}
  (\bibinfo {year} {2005})},\ \Eprint {http://arxiv.org/abs/hep-th/0411129}
  {arXiv:hep-th/0411129} \BibitemShut {NoStop}%
\bibitem [{\citenamefont {Halverson}\ \emph {et~al.}(2020)\citenamefont
  {Halverson}, \citenamefont {Plesser}, \citenamefont {Ruehle},\ and\
  \citenamefont {Tian}}]{Halverson:2019vmd}%
  \BibitemOpen
  \bibfield  {author} {\bibinfo {author} {\bibfnamefont {J.}~\bibnamefont
  {Halverson}}, \bibinfo {author} {\bibfnamefont {M.}~\bibnamefont {Plesser}},
  \bibinfo {author} {\bibfnamefont {F.}~\bibnamefont {Ruehle}}, \ and\ \bibinfo
  {author} {\bibfnamefont {J.}~\bibnamefont {Tian}},\ }\href {\doibase
  10.1103/PhysRevD.101.046010} {\bibfield  {journal} {\bibinfo  {journal}
  {Phys. Rev. D}\ }\textbf {\bibinfo {volume} {101}},\ \bibinfo {pages}
  {046010} (\bibinfo {year} {2020})},\ \Eprint
  {http://arxiv.org/abs/1911.07835} {arXiv:1911.07835 [hep-th]} \BibitemShut
  {NoStop}%
\bibitem [{\citenamefont {Hayakawa}\ and\ \citenamefont
  {Hayakawa}(1995)}]{Hayakawa1995DEGENERATIONOC}%
  \BibitemOpen
  \bibfield  {author} {\bibinfo {author} {\bibfnamefont {Y.}~\bibnamefont
  {Hayakawa}}\ and\ \bibinfo {author} {\bibfnamefont {Y.}~\bibnamefont
  {Hayakawa}}\ }(\bibinfo {year} {1995})\BibitemShut {NoStop}%
\bibitem [{\citenamefont {lung Wang}(1997)}]{Wang97onthe}%
  \BibitemOpen
  \bibfield  {author} {\bibinfo {author} {\bibfnamefont {C.}~\bibnamefont {lung
  Wang}},\ }\href@noop {} {\enquote {\bibinfo {title} {On the incompleteness of
  the weil-petersson metric along degenerations of calabi-yau manifolds},}\ }
  (\bibinfo {year} {1997})\BibitemShut {NoStop}%
\bibitem [{\citenamefont {Grassi}(1991)}]{grassiminimal}%
  \BibitemOpen
  \bibfield  {author} {\bibinfo {author} {\bibfnamefont {A.}~\bibnamefont
  {Grassi}},\ }\href {\doibase 10.1007/BF01459246} {\bibfield  {journal}
  {\bibinfo  {journal} {Mathematische Annalen}\ }\textbf {\bibinfo {volume}
  {290}},\ \bibinfo {pages} {287} (\bibinfo {year} {1991})}\BibitemShut
  {NoStop}%
\bibitem [{\citenamefont {Morrison}\ and\ \citenamefont
  {Taylor}(2012{\natexlab{a}})}]{Morrison_2012}%
  \BibitemOpen
  \bibfield  {author} {\bibinfo {author} {\bibfnamefont {D.}~\bibnamefont
  {Morrison}}\ and\ \bibinfo {author} {\bibfnamefont {W.}~\bibnamefont
  {Taylor}},\ }\href {\doibase 10.1002/prop.201200086} {\bibfield  {journal}
  {\bibinfo  {journal} {Fortschritte der Physik}\ }\textbf {\bibinfo {volume}
  {60}},\ \bibinfo {pages} {1187–1216} (\bibinfo {year}
  {2012}{\natexlab{a}})}\BibitemShut {NoStop}%
\bibitem [{\citenamefont {Morrison}\ and\ \citenamefont
  {Taylor}(2012{\natexlab{b}})}]{Morrison_20121}%
  \BibitemOpen
  \bibfield  {author} {\bibinfo {author} {\bibfnamefont {D.~R.}\ \bibnamefont
  {Morrison}}\ and\ \bibinfo {author} {\bibfnamefont {W.}~\bibnamefont
  {Taylor}},\ }\href {\doibase 10.1007/jhep01(2012)022} {\bibfield  {journal}
  {\bibinfo  {journal} {Journal of High Energy Physics}\ }\textbf {\bibinfo
  {volume} {2012}} (\bibinfo {year} {2012}{\natexlab{b}}),\
  10.1007/jhep01(2012)022}\BibitemShut {NoStop}%
\bibitem [{\citenamefont {Morrison}\ and\ \citenamefont
  {Taylor}(2012{\natexlab{c}})}]{Morrison_20122}%
  \BibitemOpen
  \bibfield  {author} {\bibinfo {author} {\bibfnamefont {D.}~\bibnamefont
  {Morrison}}\ and\ \bibinfo {author} {\bibfnamefont {W.}~\bibnamefont
  {Taylor}},\ }\href {\doibase 10.2478/s11534-012-0065-4} {\bibfield  {journal}
  {\bibinfo  {journal} {Open Physics}\ }\textbf {\bibinfo {volume} {10}}
  (\bibinfo {year} {2012}{\natexlab{c}}),\
  10.2478/s11534-012-0065-4}\BibitemShut {NoStop}%
\bibitem [{\citenamefont {Taylor}\ and\ \citenamefont
  {Wang}(2017)}]{taylor2017nontoric}%
  \BibitemOpen
  \bibfield  {author} {\bibinfo {author} {\bibfnamefont {W.}~\bibnamefont
  {Taylor}}\ and\ \bibinfo {author} {\bibfnamefont {Y.-N.}\ \bibnamefont
  {Wang}},\ }\href@noop {} {\enquote {\bibinfo {title} {Non-toric bases for
  elliptic calabi-yau threefolds and 6d f-theory vacua},}\ } (\bibinfo {year}
  {2017}),\ \Eprint {http://arxiv.org/abs/1504.07689} {arXiv:1504.07689
  [hep-th]} \BibitemShut {NoStop}%
\bibitem [{\citenamefont {{Gross}}(1993)}]{1993alg.geom..5002G}%
  \BibitemOpen
  \bibfield  {author} {\bibinfo {author} {\bibfnamefont {M.}~\bibnamefont
  {{Gross}}},\ }\href@noop {} {\bibfield  {journal} {\bibinfo  {journal} {arXiv
  e-prints}\ ,\ \bibinfo {eid} {alg-geom/9305002}} (\bibinfo {year} {1993})},\
  \Eprint {http://arxiv.org/abs/alg-geom/9305002} {arXiv:alg-geom/9305002
  [math.AG]} \BibitemShut {NoStop}%
\bibitem [{\citenamefont {Morrison}\ and\ \citenamefont
  {Taylor}(2012{\natexlab{d}})}]{morrisontaylor2012}%
  \BibitemOpen
  \bibfield  {author} {\bibinfo {author} {\bibfnamefont {D.}~\bibnamefont
  {Morrison}}\ and\ \bibinfo {author} {\bibfnamefont {W.}~\bibnamefont
  {Taylor}},\ }\href {\doibase 10.2478/s11534-012-0065-4} {\bibfield  {journal}
  {\bibinfo  {journal} {Open Physics}\ }\textbf {\bibinfo {volume} {10}}
  (\bibinfo {year} {2012}{\natexlab{d}}),\
  10.2478/s11534-012-0065-4}\BibitemShut {NoStop}%
\bibitem [{\citenamefont {Cerbo}\ and\ \citenamefont
  {Svaldi}(2019)}]{dicerbo2019birational}%
  \BibitemOpen
  \bibfield  {author} {\bibinfo {author} {\bibfnamefont {G.~D.}\ \bibnamefont
  {Cerbo}}\ and\ \bibinfo {author} {\bibfnamefont {R.}~\bibnamefont {Svaldi}},\
  }\href@noop {} {\enquote {\bibinfo {title} {Birational boundedness of low
  dimensional elliptic calabi-yau varieties with a section},}\ } (\bibinfo
  {year} {2019}),\ \Eprint {http://arxiv.org/abs/1608.02997} {arXiv:1608.02997
  [math.AG]} \BibitemShut {NoStop}%
\bibitem [{\citenamefont {BIRKAR}\ \emph {et~al.}(2010)\citenamefont {BIRKAR},
  \citenamefont {CASCINI}, \citenamefont {HACON},\ and\ \citenamefont
  {McKERNAN}}]{10.2307/40587279}%
  \BibitemOpen
  \bibfield  {author} {\bibinfo {author} {\bibfnamefont {C.}~\bibnamefont
  {BIRKAR}}, \bibinfo {author} {\bibfnamefont {P.}~\bibnamefont {CASCINI}},
  \bibinfo {author} {\bibfnamefont {C.~D.}\ \bibnamefont {HACON}}, \ and\
  \bibinfo {author} {\bibfnamefont {J.}~\bibnamefont {McKERNAN}},\ }\href
  {http://www.jstor.org/stable/40587279} {\bibfield  {journal} {\bibinfo
  {journal} {Journal of the American Mathematical Society}\ }\textbf {\bibinfo
  {volume} {23}},\ \bibinfo {pages} {405} (\bibinfo {year} {2010})}\BibitemShut
  {NoStop}%
\bibitem [{\citenamefont {Harvey}\ and\ \citenamefont
  {Lawson}(1982)}]{10.1007/BF02392726}%
  \BibitemOpen
  \bibfield  {author} {\bibinfo {author} {\bibfnamefont {R.}~\bibnamefont
  {Harvey}}\ and\ \bibinfo {author} {\bibfnamefont {H.~B.}\ \bibnamefont
  {Lawson}},\ }\href {\doibase 10.1007/BF02392726} {\bibfield  {journal}
  {\bibinfo  {journal} {Acta Mathematica}\ }\textbf {\bibinfo {volume} {148}},\
  \bibinfo {pages} {47 } (\bibinfo {year} {1982})}\BibitemShut {NoStop}%
\bibitem [{\citenamefont {Halverson}\ and\ \citenamefont
  {Ruehle}(2019)}]{Halverson:2018cio}%
  \BibitemOpen
  \bibfield  {author} {\bibinfo {author} {\bibfnamefont {J.}~\bibnamefont
  {Halverson}}\ and\ \bibinfo {author} {\bibfnamefont {F.}~\bibnamefont
  {Ruehle}},\ }\href {\doibase 10.1103/PhysRevD.99.046015} {\bibfield
  {journal} {\bibinfo  {journal} {Phys. Rev. D}\ }\textbf {\bibinfo {volume}
  {99}},\ \bibinfo {pages} {046015} (\bibinfo {year} {2019})},\ \Eprint
  {http://arxiv.org/abs/1809.08279} {arXiv:1809.08279 [hep-th]} \BibitemShut
  {NoStop}%
\bibitem [{\citenamefont {Demirtas}\ \emph {et~al.}(2020)\citenamefont
  {Demirtas}, \citenamefont {Long}, \citenamefont {McAllister},\ and\
  \citenamefont {Stillman}}]{Demirtas:2018akl}%
  \BibitemOpen
  \bibfield  {author} {\bibinfo {author} {\bibfnamefont {M.}~\bibnamefont
  {Demirtas}}, \bibinfo {author} {\bibfnamefont {C.}~\bibnamefont {Long}},
  \bibinfo {author} {\bibfnamefont {L.}~\bibnamefont {McAllister}}, \ and\
  \bibinfo {author} {\bibfnamefont {M.}~\bibnamefont {Stillman}},\ }\href
  {\doibase 10.1007/JHEP04(2020)138} {\bibfield  {journal} {\bibinfo  {journal}
  {JHEP}\ }\textbf {\bibinfo {volume} {04}},\ \bibinfo {pages} {138} (\bibinfo
  {year} {2020})},\ \Eprint {http://arxiv.org/abs/1808.01282} {arXiv:1808.01282
  [hep-th]} \BibitemShut {NoStop}%
\bibitem [{\citenamefont {Morrison}\ and\ \citenamefont
  {Vafa}(1996{\natexlab{a}})}]{Morrison:1996na}%
  \BibitemOpen
  \bibfield  {author} {\bibinfo {author} {\bibfnamefont {D.~R.}\ \bibnamefont
  {Morrison}}\ and\ \bibinfo {author} {\bibfnamefont {C.}~\bibnamefont
  {Vafa}},\ }\href {\doibase 10.1016/0550-3213(96)00242-8} {\bibfield
  {journal} {\bibinfo  {journal} {Nucl. Phys. B}\ }\textbf {\bibinfo {volume}
  {473}},\ \bibinfo {pages} {74} (\bibinfo {year} {1996}{\natexlab{a}})},\
  \Eprint {http://arxiv.org/abs/hep-th/9602114} {arXiv:hep-th/9602114}
  \BibitemShut {NoStop}%
\bibitem [{\citenamefont {Morrison}\ and\ \citenamefont
  {Vafa}(1996{\natexlab{b}})}]{Morrison:1996pp}%
  \BibitemOpen
  \bibfield  {author} {\bibinfo {author} {\bibfnamefont {D.~R.}\ \bibnamefont
  {Morrison}}\ and\ \bibinfo {author} {\bibfnamefont {C.}~\bibnamefont
  {Vafa}},\ }\href {\doibase 10.1016/0550-3213(96)00369-0} {\bibfield
  {journal} {\bibinfo  {journal} {Nucl. Phys. B}\ }\textbf {\bibinfo {volume}
  {476}},\ \bibinfo {pages} {437} (\bibinfo {year} {1996}{\natexlab{b}})},\
  \Eprint {http://arxiv.org/abs/hep-th/9603161} {arXiv:hep-th/9603161}
  \BibitemShut {NoStop}%
\bibitem [{\citenamefont {Morrison}\ and\ \citenamefont
  {Taylor}(2012{\natexlab{e}})}]{Morrison:2012js}%
  \BibitemOpen
  \bibfield  {author} {\bibinfo {author} {\bibfnamefont {D.~R.}\ \bibnamefont
  {Morrison}}\ and\ \bibinfo {author} {\bibfnamefont {W.}~\bibnamefont
  {Taylor}},\ }\href {\doibase 10.1002/prop.201200086} {\bibfield  {journal}
  {\bibinfo  {journal} {Fortsch. Phys.}\ }\textbf {\bibinfo {volume} {60}},\
  \bibinfo {pages} {1187} (\bibinfo {year} {2012}{\natexlab{e}})},\ \Eprint
  {http://arxiv.org/abs/1204.0283} {arXiv:1204.0283 [hep-th]} \BibitemShut
  {NoStop}%
\bibitem [{\citenamefont {Grassi}\ \emph {et~al.}(2015)\citenamefont {Grassi},
  \citenamefont {Halverson}, \citenamefont {Shaneson},\ and\ \citenamefont
  {Taylor}}]{Grassi:2014zxa}%
  \BibitemOpen
  \bibfield  {author} {\bibinfo {author} {\bibfnamefont {A.}~\bibnamefont
  {Grassi}}, \bibinfo {author} {\bibfnamefont {J.}~\bibnamefont {Halverson}},
  \bibinfo {author} {\bibfnamefont {J.}~\bibnamefont {Shaneson}}, \ and\
  \bibinfo {author} {\bibfnamefont {W.}~\bibnamefont {Taylor}},\ }\href
  {\doibase 10.1007/JHEP01(2015)086} {\bibfield  {journal} {\bibinfo  {journal}
  {JHEP}\ }\textbf {\bibinfo {volume} {01}},\ \bibinfo {pages} {086} (\bibinfo
  {year} {2015})},\ \Eprint {http://arxiv.org/abs/1409.8295} {arXiv:1409.8295
  [hep-th]} \BibitemShut {NoStop}%
\bibitem [{\citenamefont {Morrison}\ and\ \citenamefont
  {Taylor}(2015)}]{Morrison:2014lca}%
  \BibitemOpen
  \bibfield  {author} {\bibinfo {author} {\bibfnamefont {D.~R.}\ \bibnamefont
  {Morrison}}\ and\ \bibinfo {author} {\bibfnamefont {W.}~\bibnamefont
  {Taylor}},\ }\href {\doibase 10.1007/JHEP05(2015)080} {\bibfield  {journal}
  {\bibinfo  {journal} {JHEP}\ }\textbf {\bibinfo {volume} {05}},\ \bibinfo
  {pages} {080} (\bibinfo {year} {2015})},\ \Eprint
  {http://arxiv.org/abs/1412.6112} {arXiv:1412.6112 [hep-th]} \BibitemShut
  {NoStop}%
\bibitem [{\citenamefont {Halverson}\ and\ \citenamefont
  {Taylor}(2015)}]{Halverson:2015jua}%
  \BibitemOpen
  \bibfield  {author} {\bibinfo {author} {\bibfnamefont {J.}~\bibnamefont
  {Halverson}}\ and\ \bibinfo {author} {\bibfnamefont {W.}~\bibnamefont
  {Taylor}},\ }\href {\doibase 10.1007/JHEP09(2015)086} {\bibfield  {journal}
  {\bibinfo  {journal} {JHEP}\ }\textbf {\bibinfo {volume} {09}},\ \bibinfo
  {pages} {086} (\bibinfo {year} {2015})},\ \Eprint
  {http://arxiv.org/abs/1506.03204} {arXiv:1506.03204 [hep-th]} \BibitemShut
  {NoStop}%
\bibitem [{\citenamefont {Halverson}(2017)}]{Halverson:2016vwx}%
  \BibitemOpen
  \bibfield  {author} {\bibinfo {author} {\bibfnamefont {J.}~\bibnamefont
  {Halverson}},\ }\href {\doibase 10.1016/j.nuclphysb.2017.02.014} {\bibfield
  {journal} {\bibinfo  {journal} {Nucl. Phys. B}\ }\textbf {\bibinfo {volume}
  {919}},\ \bibinfo {pages} {267} (\bibinfo {year} {2017})},\ \Eprint
  {http://arxiv.org/abs/1603.01639} {arXiv:1603.01639 [hep-th]} \BibitemShut
  {NoStop}%
\bibitem [{\citenamefont {Taylor}\ and\ \citenamefont
  {Wang}(2016)}]{Taylor:2015ppa}%
  \BibitemOpen
  \bibfield  {author} {\bibinfo {author} {\bibfnamefont {W.}~\bibnamefont
  {Taylor}}\ and\ \bibinfo {author} {\bibfnamefont {Y.-N.}\ \bibnamefont
  {Wang}},\ }\href {\doibase 10.1007/JHEP01(2016)137} {\bibfield  {journal}
  {\bibinfo  {journal} {JHEP}\ }\textbf {\bibinfo {volume} {01}},\ \bibinfo
  {pages} {137} (\bibinfo {year} {2016})},\ \Eprint
  {http://arxiv.org/abs/1510.04978} {arXiv:1510.04978 [hep-th]} \BibitemShut
  {NoStop}%
\bibitem [{\citenamefont {Taylor}\ and\ \citenamefont
  {Wang}(2018)}]{Taylor:2017yqr}%
  \BibitemOpen
  \bibfield  {author} {\bibinfo {author} {\bibfnamefont {W.}~\bibnamefont
  {Taylor}}\ and\ \bibinfo {author} {\bibfnamefont {Y.-N.}\ \bibnamefont
  {Wang}},\ }\href {\doibase 10.1007/JHEP01(2018)111} {\bibfield  {journal}
  {\bibinfo  {journal} {JHEP}\ }\textbf {\bibinfo {volume} {01}},\ \bibinfo
  {pages} {111} (\bibinfo {year} {2018})},\ \Eprint
  {http://arxiv.org/abs/1710.11235} {arXiv:1710.11235 [hep-th]} \BibitemShut
  {NoStop}%
\bibitem [{\citenamefont {Taylor}\ and\ \citenamefont
  {Wang}(2015)}]{Taylor:2015xtz}%
  \BibitemOpen
  \bibfield  {author} {\bibinfo {author} {\bibfnamefont {W.}~\bibnamefont
  {Taylor}}\ and\ \bibinfo {author} {\bibfnamefont {Y.-N.}\ \bibnamefont
  {Wang}},\ }\href {\doibase 10.1007/JHEP12(2015)164} {\bibfield  {journal}
  {\bibinfo  {journal} {JHEP}\ }\textbf {\bibinfo {volume} {12}},\ \bibinfo
  {pages} {164} (\bibinfo {year} {2015})},\ \Eprint
  {http://arxiv.org/abs/1511.03209} {arXiv:1511.03209 [hep-th]} \BibitemShut
  {NoStop}%
\bibitem [{\citenamefont {Halverson}\ \emph
  {et~al.}(2017{\natexlab{b}})\citenamefont {Halverson}, \citenamefont
  {Nelson},\ and\ \citenamefont {Ruehle}}]{Halverson:2016nfq}%
  \BibitemOpen
  \bibfield  {author} {\bibinfo {author} {\bibfnamefont {J.}~\bibnamefont
  {Halverson}}, \bibinfo {author} {\bibfnamefont {B.~D.}\ \bibnamefont
  {Nelson}}, \ and\ \bibinfo {author} {\bibfnamefont {F.}~\bibnamefont
  {Ruehle}},\ }\href {\doibase 10.1103/PhysRevD.95.043527} {\bibfield
  {journal} {\bibinfo  {journal} {Phys. Rev. D}\ }\textbf {\bibinfo {volume}
  {95}},\ \bibinfo {pages} {043527} (\bibinfo {year} {2017}{\natexlab{b}})},\
  \Eprint {http://arxiv.org/abs/1609.02151} {arXiv:1609.02151 [hep-ph]}
  \BibitemShut {NoStop}%
\bibitem [{\citenamefont {Halverson}\ \emph {et~al.}(2018)\citenamefont
  {Halverson}, \citenamefont {Nelson}, \citenamefont {Ruehle},\ and\
  \citenamefont {Salinas}}]{Halverson:2018olu}%
  \BibitemOpen
  \bibfield  {author} {\bibinfo {author} {\bibfnamefont {J.}~\bibnamefont
  {Halverson}}, \bibinfo {author} {\bibfnamefont {B.~D.}\ \bibnamefont
  {Nelson}}, \bibinfo {author} {\bibfnamefont {F.}~\bibnamefont {Ruehle}}, \
  and\ \bibinfo {author} {\bibfnamefont {G.}~\bibnamefont {Salinas}},\ }\href
  {\doibase 10.1103/PhysRevD.98.043502} {\bibfield  {journal} {\bibinfo
  {journal} {Phys. Rev. D}\ }\textbf {\bibinfo {volume} {98}},\ \bibinfo
  {pages} {043502} (\bibinfo {year} {2018})},\ \Eprint
  {http://arxiv.org/abs/1805.06011} {arXiv:1805.06011 [hep-ph]} \BibitemShut
  {NoStop}%
\bibitem [{\citenamefont {Halverson}\ \emph
  {et~al.}(2019{\natexlab{a}})\citenamefont {Halverson}, \citenamefont {Long},
  \citenamefont {Nelson},\ and\ \citenamefont {Salinas}}]{Halverson:2019kna}%
  \BibitemOpen
  \bibfield  {author} {\bibinfo {author} {\bibfnamefont {J.}~\bibnamefont
  {Halverson}}, \bibinfo {author} {\bibfnamefont {C.}~\bibnamefont {Long}},
  \bibinfo {author} {\bibfnamefont {B.}~\bibnamefont {Nelson}}, \ and\ \bibinfo
  {author} {\bibfnamefont {G.}~\bibnamefont {Salinas}},\ }\href {\doibase
  10.1103/PhysRevD.99.086014} {\bibfield  {journal} {\bibinfo  {journal} {Phys.
  Rev. D}\ }\textbf {\bibinfo {volume} {99}},\ \bibinfo {pages} {086014}
  (\bibinfo {year} {2019}{\natexlab{a}})},\ \Eprint
  {http://arxiv.org/abs/1903.04495} {arXiv:1903.04495 [hep-th]} \BibitemShut
  {NoStop}%
\bibitem [{\citenamefont {Halverson}\ \emph
  {et~al.}(2019{\natexlab{b}})\citenamefont {Halverson}, \citenamefont {Long},
  \citenamefont {Nelson},\ and\ \citenamefont {Salinas}}]{Halverson:2019cmy}%
  \BibitemOpen
  \bibfield  {author} {\bibinfo {author} {\bibfnamefont {J.}~\bibnamefont
  {Halverson}}, \bibinfo {author} {\bibfnamefont {C.}~\bibnamefont {Long}},
  \bibinfo {author} {\bibfnamefont {B.}~\bibnamefont {Nelson}}, \ and\ \bibinfo
  {author} {\bibfnamefont {G.}~\bibnamefont {Salinas}},\ }\href {\doibase
  10.1103/PhysRevD.100.106010} {\bibfield  {journal} {\bibinfo  {journal}
  {Phys. Rev. D}\ }\textbf {\bibinfo {volume} {100}},\ \bibinfo {pages}
  {106010} (\bibinfo {year} {2019}{\natexlab{b}})},\ \Eprint
  {http://arxiv.org/abs/1909.05257} {arXiv:1909.05257 [hep-th]} \BibitemShut
  {NoStop}%
\bibitem [{\citenamefont {Halverson}\ \emph {et~al.}(2021)\citenamefont
  {Halverson}, \citenamefont {Long}, \citenamefont {Maiti}, \citenamefont
  {Nelson},\ and\ \citenamefont {Salinas}}]{Halverson:2020xpg}%
  \BibitemOpen
  \bibfield  {author} {\bibinfo {author} {\bibfnamefont {J.}~\bibnamefont
  {Halverson}}, \bibinfo {author} {\bibfnamefont {C.}~\bibnamefont {Long}},
  \bibinfo {author} {\bibfnamefont {A.}~\bibnamefont {Maiti}}, \bibinfo
  {author} {\bibfnamefont {B.}~\bibnamefont {Nelson}}, \ and\ \bibinfo {author}
  {\bibfnamefont {G.}~\bibnamefont {Salinas}},\ }\href {\doibase
  10.1007/JHEP05(2021)154} {\bibfield  {journal} {\bibinfo  {journal} {JHEP}\
  }\textbf {\bibinfo {volume} {05}},\ \bibinfo {pages} {154} (\bibinfo {year}
  {2021})},\ \Eprint {http://arxiv.org/abs/2012.04071} {arXiv:2012.04071
  [hep-ph]} \BibitemShut {NoStop}%
\bibitem [{\citenamefont {Kreuzer}\ and\ \citenamefont
  {Skarke}(1998)}]{Kreuzer:1998vb}%
  \BibitemOpen
  \bibfield  {author} {\bibinfo {author} {\bibfnamefont {M.}~\bibnamefont
  {Kreuzer}}\ and\ \bibinfo {author} {\bibfnamefont {H.}~\bibnamefont
  {Skarke}},\ }\href {\doibase 10.4310/ATMP.1998.v2.n4.a5} {\bibfield
  {journal} {\bibinfo  {journal} {Adv. Theor. Math. Phys.}\ }\textbf {\bibinfo
  {volume} {2}},\ \bibinfo {pages} {853} (\bibinfo {year} {1998})},\ \Eprint
  {http://arxiv.org/abs/hep-th/9805190} {arXiv:hep-th/9805190} \BibitemShut
  {NoStop}%
\bibitem [{\citenamefont {Kollar}\ and\ \citenamefont
  {Mori}(1998)}]{kollar_mori_1998}%
  \BibitemOpen
  \bibfield  {author} {\bibinfo {author} {\bibfnamefont {J.}~\bibnamefont
  {Kollar}}\ and\ \bibinfo {author} {\bibfnamefont {S.}~\bibnamefont {Mori}},\
  }\href {\doibase 10.1017/CBO9780511662560} {\emph {\bibinfo {title}
  {Birational Geometry of Algebraic Varieties}}},\ Cambridge Tracts in
  Mathematics\ (\bibinfo  {publisher} {Cambridge University Press},\ \bibinfo
  {year} {1998})\BibitemShut {NoStop}%
\bibitem [{\citenamefont {(https://mathoverflow.net/users/10076/sc3a1ndor
  kovc3a1cs)}()}]{54125}%
  \BibitemOpen
  \bibfield  {author} {\bibinfo {author} {\bibfnamefont {S.~K.}\ \bibnamefont
  {(https://mathoverflow.net/users/10076/sc3a1ndor kovc3a1cs)}},\ }\href
  {https://mathoverflow.net/q/54125} {\enquote {\bibinfo {title} {When does
  f-nef imply nef (after twisting?)},}\ }\bibinfo {howpublished}
  {MathOverflow},\ \bibinfo {note} {uRL:https://mathoverflow.net/q/54125
  (version: 2017-06-26)},\ \Eprint
  {http://arxiv.org/abs/https://mathoverflow.net/q/54125}
  {https://mathoverflow.net/q/54125} \BibitemShut {NoStop}%
\bibitem [{\citenamefont {Felgueiras}(2008)}]{ampcone}%
  \BibitemOpen
  \bibfield  {author} {\bibinfo {author} {\bibfnamefont {O.}~\bibnamefont
  {Felgueiras}},\ }\href@noop {} {\emph {\bibinfo {title} {The Ample Cone of a
  Morphism}}}\ (\bibinfo {year} {2008})\BibitemShut {NoStop}%
\bibitem [{\citenamefont {Eisenbud}\ and\ \citenamefont
  {Harris}(2016)}]{eisenbud_harris_2016}%
  \BibitemOpen
  \bibfield  {author} {\bibinfo {author} {\bibfnamefont {D.}~\bibnamefont
  {Eisenbud}}\ and\ \bibinfo {author} {\bibfnamefont {J.}~\bibnamefont
  {Harris}},\ }\href {\doibase 10.1017/CBO9781139062046} {\emph {\bibinfo
  {title} {3264 and All That: A Second Course in Algebraic Geometry}}}\
  (\bibinfo  {publisher} {Cambridge University Press},\ \bibinfo {year}
  {2016})\BibitemShut {NoStop}%
\bibitem [{\citenamefont {Iskovskikh}\ and\ \citenamefont
  {Prokhorov}(1999)}]{fanovarieties}%
  \BibitemOpen
  \bibfield  {author} {\bibinfo {author} {\bibfnamefont {V.}~\bibnamefont
  {Iskovskikh}}\ and\ \bibinfo {author} {\bibfnamefont {Y.}~\bibnamefont
  {Prokhorov}},\ }\href@noop {} {\emph {\bibinfo {title} {Fano varieties.
  Algebraic geometry. V.}}}\ (\bibinfo {year} {1999})\BibitemShut {NoStop}%
\end{thebibliography}%
\end{document}